\documentclass{amsart}

\usepackage{latexsym}
\usepackage{amsmath}
\usepackage{amsthm}
\usepackage{amssymb}
\usepackage{mathrsfs}
\usepackage{wasysym,color}
\usepackage{graphicx}

\usepackage{amsfonts}

\usepackage{graphicx}%
\newtheorem{theorem}{\indent Theorem}[section]
\newtheorem{proposition}[theorem]{\indent Proposition}
\newtheorem{definition}[theorem]{\indent Definition}
\newtheorem{lemma}[theorem]{\indent Lemma}
\newtheorem{remark}[theorem]{\indent Remark}
\newtheorem{corollary}[theorem]{\indent Corollary}
\newtheorem{observation}[theorem]{\indent Observation}
\numberwithin{equation}{section}

\newcommand{\R}{\mathbb{R}}

\title{Quasi-geostrophic equation in $\mathbb{R}^2$}
\author{Tomasz Dlotko, Maria B. Kania}\thanks{Corresponding author: T. Dlotko, Institute of Mathematics, Silesian University, Telephone: (+48)322582976}
\address[Tomasz Dlotko, Maria B. Kania]{Institute of Mathematics, Silesian University, 40-007 Katowice, Bankowa 14,  Poland}
\email{tdlotko@math.us.edu.pl, mkania@math.us.edu.pl}
\author{Chunyou Sun}
\address[Chunyou Sun]{School of Mathematics and Statistics, Lanzhou University, 730000, Lanzhou, P.R. China}
\email{sunchy@lzu.edu.cn}

\begin{document}

\date{}
\begin{abstract}
Solvability of Cauchy's problem in $\mathbb{R}^2$ for subcritical quasi-geostrophic equation is discussed here in two phase spaces; $L^p(\R^2)$ with $p> \frac{2}{2\alpha-1}$ and $H^s(\R^2)$ with $s>1$. A solution to that equation in {\it critical case} is obtained next as a limit of the $H^s$-solutions to subcritical equations when the exponent $\alpha$ of   $(-\Delta)^\alpha$ tends to $\frac{1}{2}^+$. Such idea seems to be new in the literature. Existence of the global attractor in subcritical case is discussed in the paper. In section \ref{addedsect} we also discuss solvability of the critical problem with Dirichlet boundary condition in bounded domain $\Omega \subset \R^2$, when $\| \theta_0 \|_{L^\infty(\Omega)}$ is small.      
\end{abstract}
\maketitle

\vspace{-1cm}
\section{Introduction}
The dissipative quasi-geostrophic equation considered here has the form:
\begin{equation}\label{Q-gequation}
\begin{split}
&\theta_t + u \cdot \nabla\theta + \kappa(-\Delta)^\alpha \theta = f, \quad x \in \R^2,\, t>0,  \\
&\theta(0,x) = \theta_0(x),
\end{split}
\end{equation}
where $\theta$ represents the potential temperature, $\kappa > 0$ is a diffusivity coefficient, $\alpha \in [\frac{1}{2}, 1]$ a fractional exponent,
and $u = (u_1,u_2)$ is the {\it velocity field} determined by $\theta$ through the relation:
\begin{equation}\label{1.2}
u = \bigl(-\frac{\partial \psi}{\partial x_2}, \frac{\partial \psi}{\partial x_1}\bigr), \quad  \text{where} \quad (-\Delta)^{\frac{1}{2}} \psi = -\theta,
\end{equation}
or, in a more explicit way,
\begin{equation}\label{1.3}
u = \bigl(- R_2\theta, R_1\theta \bigr),
\end{equation}
where $R_i, i=1,2$ are the {\it Riesz transforms}.

This paper is devoted to the global in time solvability and properties of solutions to the Cauchy problem \eqref{Q-gequation}. There is a large literature devoted to that problem published in the last 20 years; compare \cite{C-C, C-C1, K-N-V, T-M, W1, W2} for more references. The basic approach was to obtain a {\it weak solution} to \eqref{Q-gequation} using the {\it viscosity technique} (e.g. \cite{Li}), that means adding the viscosity term $\epsilon \Delta \theta$ to the right hand side of the equation, solving the regularized problem and  letting $\epsilon \to 0^+$. Our approach is different. We consider first the {\it family of subcritical problems} \eqref{Q-gequation} with $\alpha \in (\frac{1}{2}, 1]$, which can be treated in the framework of \cite{HE, C-D} as semilinear parabolic equations. Thanks to a Maximum Principle valid for \eqref{Q-gequation} (Lemma \ref{MaxPrinc}), we have a {\it uniform in} $\alpha \in (\frac{1}{2}, 1]$ estimates of the solutions to that subcritical problems in $L^p(\R^2), 1\leq p\leq +\infty$. Letting $\alpha \to \frac{1}{2}^+$ over a sequence of regular $H^{2\alpha^-+s}(\R^2), s>1$, solutions $\theta^\alpha$ to \eqref{Q-gequation}, this property allows us to introduce in Theorem \ref{critth} a {\it 'weak $L^p$ solution'} of the limiting {\it critical problem} \eqref{Q-gequation}. Such approach seems to be new in the literature. Our considerations relates most closely to the J. Wu papers \cite{W1, W2}, using however another approach of semilinear parabolic equations with sectorial operators \cite{HE, C-D}. 

\subsection{Description of the results.} The quasi-geostrophic equation \eqref{Q-gequation} considered in this paper is a challenging problem to study; compare \cite{KI, K-N-V, W1, W2}. Probably because of that many papers devoted to it were published only very recently; some of them are listed in the references.    

There are different possible choices of the phase space for that problem. Following the considerations of J. Wu \cite{W1, W2}, we chose $L^p(\R^2)$ with $p$ sufficiently large and $H^s(\R^2)$ with $s>1$ as the {\it base spaces} (in which the equation is fulfilled) for our problem; see also \cite{K-N-V}. Our aim was to include, in a subcritical case of exponent $\alpha\in (\frac{1}{2}, 1]$, the problem \eqref{Q-gequation} into the frame of semilinear parabolic equations with sectorial operator (see \cite{HE, C-D}). This offers a simple but formalized proof of the local solvability and regularity in subcritical case. The presented in section \ref{critical} approach to critical nonlinearity is new here. However, using weak compactness of bounded sets in reflexive Banach space as a tool for getting convergence to a weak solution of the critical problem, we will not be able to show that the limit of the nonlinearities of subcritical problems equals to nonlinearity of the limiting critical problem. The existing uniform in $\alpha \in (\frac{1}{2}, 1]$ a priori estimates are too weak to guarantee such property. However, they work well in case of all the linear components in the equation. The main result obtained in that direction is formulated in Theorem \ref{critth} of section \ref{critical}, where we introduce the notion of the {\it weak $L^p$ solution} to the critical equation \eqref{Q-gequation}. It is followed by four technical observations used in the proof of that theorem; the main technical result there seems to be Lemma \ref{lem6.2}, vital when letting $\alpha\to \frac{1}{2}^+$ in \eqref{Q-gequation}.   

Another result reported in our paper is existence of the global attractors for a variant of subcritical problem (with added linear damping term $\lambda \theta$, needed in case of $\R^2$) in the phase space $H^s(\R^2), s>1$. Asymptotic compactness, the most difficult point there is obtained using the {\it tail estimates technique}; see \cite{WA} as a source reference for that technique. The final section of the paper contains some known but interesting facts used in the main body of the paper. We are using in the paper the technique originated in our recent publications \cite{DL, D-S, D-K-S}. 

The final section \ref{addedsect}, following the earlier considerations of \cite{C-C-W}, is devoted to the global in time solvability of the critical Quasi-geostrophic equation with Dirichlet boundary condition in a bounded domain $\Omega \subset \R^2$. Using technique similar as in earlier sections we prove there the convergence, in a sense of $H^{-\frac{1}{2}}(\Omega)$ space, of the sequence $\theta^\alpha$ of solutions to subcritical problems  to the solution $\theta$ of the limit critical problem.     

{\it Notation.} We are using standard notation for Sobolev spaces and Riesz potentials. For $r \in \mathbb{R}$, let $r^-$ denotes a number strictly less than $r$ but close to it. Similarly, $r^+ >r$ and $r^+$ close to $r$. When needed for clarity of the presentation, we mark the dependence of the solution $\theta$ of \eqref{Q-gequation} on $\alpha \in (\frac{1}{2}, 1]$, calling it $\theta^\alpha$. Same time we do not mark explicitly the dependence of $u$ on $\alpha$, since $u$ always stays next to $\theta^\alpha$.    

\section{Solvability of subcritical \eqref{Q-gequation}, $\alpha \in (\frac{1}{2}, 1]$, in $W^{2\alpha^-,p}(\R^2)$.}\label{locsolv}
\subsection{Formulation of the problem and its local solvability}\label{sub2.2}
Our first task is the local in time solvability of the subcritical problem \eqref{Q-gequation} when equation is treated in the {\it base space} $X:=L^p(\R^2)$. We will use a standard approach of Dan Henry \cite{HE} to semilinear 'parabolic' equations. To work with a {\it sectorial positive} operator (see \cite{HE, C-D}), we will rewrite \eqref{Q-gequation} in an equivalent form:
\begin{equation}\label{2.3}
\begin{split}
&\theta_t + u \cdot \nabla\theta + \kappa(-\Delta)^\alpha \theta + \kappa\theta = f + \kappa\theta, \quad x \in \R^2,\, t>0,  \\
&\theta(0,x) = \theta_0(x).
\end{split}
\end{equation}
Define $A_\alpha := \kappa[(-\Delta)^\alpha +I]$, $\alpha \in (\frac{1}{2}, 1]$, where $(-\Delta)^\alpha$ is the fractional Laplacian. Also, setting
\begin{equation}\label{2.4}
F(\theta) = R_2 \theta \frac{\partial \theta}{\partial x_1} - R_1 \theta \frac{\partial \theta}{\partial x_2} + f + \kappa\theta,
\end{equation}
the problem \eqref{2.3} will be written formally as
\begin{equation}
\begin{split}
&\theta_t + A_\alpha \theta = F(\theta), \  t>0,  \\
&\theta(0) = \theta_0,
\end{split}
\end{equation}
which is an abstract 'parabolic' equation with sectorial positive operator. To assure that the nonlinearity $F$ is bounded and Lipschitz continuous on bounded sets as a map from the {\it phase space} $X^\beta:=W^{2\alpha^-, p}(\R^2)$ to $X$  ($X^\beta$ - domain of the $\beta$ fractional power of the sectorial operator $A_\alpha$), we need to have a sufficiently large value of $p$. More precisely, we assume the following condition known e.g from the references \cite{T-M, W1}:
\begin{equation}\label{asonF}
2\alpha - \frac{2}{p} > 1.
\end{equation}
\begin{theorem}\label{localexistence}
When the condition \eqref{asonF} holds, $f\in L^p(\R^2)$ and $\theta_0 \in W^{2\alpha^-,p}(\R^2)$, then there exists a unique local in time mild solution $\theta(t)$ to the subcritical problem \eqref{Q-gequation} considered on the phase space $W^{2\alpha^-,p}(\R^2)$. Moreover, %due to \cite{HE, C-D}:
\begin{equation*}
\theta \in C((0,\tau); W^{2\alpha,p}(\R^2)) \cap C([0,\tau); W^{2\alpha^-,p}(\R^2)), \  \theta_t \in C((0,\tau); W^{2\gamma,p}(\R^2)),
\end{equation*}
with arbitrary $\gamma < \alpha^-$. Here $\tau> 0$ is the 'live time' of that local in time solution. Moreover, the {\em Cauchy formula} is satisfied:
\begin{equation*}
\theta(t) = e^{-A_\alpha t} \theta_0 + \int_0^t e^{-A_\alpha(t-s)} F(\theta(s)) ds, \  t \in [0, \tau),
\end{equation*}
where $e^{-A_\alpha t}$ denotes the linear semigroup corresponding to the operator $A_\alpha := \kappa[(-\Delta)^\alpha +I]$ in $L^p(\R^2)$.
\end{theorem}
\begin{proof}
 To guarantee the local solvability (see e.g. \cite{C-D}, p. 55) we need to have the local Lipschitz condition:
\begin{equation}\label{loclip}
\forall_{\theta_1, \theta_2 \in B_{X^\beta}(r)} \exists_{L(r)}  \| F(\theta_1) - F(\theta_2) \|_X \leq L(r) \| \theta_1-\theta_2 \|_{X^\beta},
\end{equation}
where $B_{X^\beta}(r)$ denotes an open ball in $X^\beta$  centered at zero of radius r. 
Using \eqref{1.3} and, for $\theta_1, \theta_2 \in B(r)$ denoting: $u_1=(-R_2 \theta_1, R_1 \theta_1), u_2=(-R_2 \theta_2, R_1 \theta_2)$, we obtain
\begin{equation}\label{2.6}
\begin{split}
\| F(\theta_1) - F(\theta_2) \|_{L^p(\R^2)} &\leq \kappa \| \theta_1 -\theta_2 \|_{L^p(\R^2)} + \| R_2(\theta_1 -\theta_2) \frac{\partial \theta_1}{\partial x_1} + R_2 \theta_2 \frac{\partial(\theta_1-\theta_2)}{\partial x_1}\|_{L^p(\R^2)} \\
&+ \|  R_1(\theta_1-\theta_2) \frac{\partial \theta_1}{\partial x_2} + R_1 \theta_2 \frac{\partial(\theta_1-\theta_2)}{\partial x_2} \|_{L^p(\R^2)}.
\end{split}
\end{equation}
Next, using H\"older inequality and \eqref{Rieszobs}, we  estimate the second term above as follows:
\begin{equation}\label{2.7}
\| R_2(\theta_1 -\theta_2) \frac{\partial \theta_1}{\partial x_1}\|_{L^p(\R^2)} \leq \| R_2(\theta_1 -\theta_2) \|_{L^{2p}(\R^2)} \| \frac{\partial \theta_1}{\partial x_1}\|_{L^{2p}(\R^2)}  
 \leq C \| \theta_1 -\theta_2 \|_{L^{2p}(\R^2)} \| \theta_1\|_{W^{1,2p}(\R^2)}.
\end{equation}
Note that the embeddings $W^{2\alpha^-,p}(\R^2) \subset W^{1,2p}(\R^2)\subset L^{2p}(\R^2)$ are valid since the condition $2\alpha^- - \frac{1}{p} \geq 1$ is satisfied.
Estimating the others components in \eqref{2.6} analogously we conclude:
\begin{equation*}
\| F(\theta_1) - F(\theta_2) \|_{L^p(\R^2)} \leq const\bigl( \| \theta_1 \|_{W^{2\alpha^-,p}(\R^2)}, \| \theta_2 \|_{W^{2\alpha^-,p}(\R^2)} \bigr) \| \theta_1 - \theta_2 \|_{W^{2\alpha^-,p}(\R^2)},
\end{equation*}
which is the required local Lipschitz condition. 
\end{proof}

\begin{remark}
It is seen from the estimate \eqref{2.7} that the quasi-geostrophic equation \eqref{Q-gequation} fall into the class of equations with {\em quadratic nonlinearity}. This is typical also for another equations originated in fluid dynamics as the Burgers equation (see e.g. \cite{B-K-W}) or the celebrated Navier-Stokes equation where a significant part of the considerations is devoted (see e.g. \cite[Chapt. III]{TE1}) to the studies of the {\it bi-linear} or {\it tri-linear forms} connected with the nonlinearity there. 
In our problem, at the first view, the a priori estimate following from the Maximum Principle seems to be sufficient to assure the global in time solvability of such type problems. But usually the detail considerations leading to such conclusion require more attention.   
\end{remark}

\subsection{Global solvability.} Having already the local in time solution of \eqref{Q-gequation}, $\alpha\in (\frac{1}{2}, 1]$, to guarantee its global extendibility we need to have suitable {\it a priori estimates}. Such role in case of the viscous quasi-geostrophic equation \eqref{Q-gequation} is played by the Maximum Principle:

\begin{lemma}\label{MaxPrinc}
For arbitrary $q \in [2,\infty)$, and a sufficiently regular solution of \eqref{Q-gequation}, if $f$ is non-zero and in $L^q(\R^2)$  then
\begin{equation}\label{maxp11}
\|\theta(t,\cdot)\|^q_{L^q(\R^2)} \leq \| \theta_0 \|^q_{L^q(\R^2)} e^{(q-1)t} + \frac{e^{(q-1)t}-1}{q-1} \|f\|_{L^q(\R^2)}^q.
\end{equation}
When $f=0$, the corresponding estimate will take the form:
\begin{equation}\label{maxp}
\|\theta(t,\cdot)\|_{L^q(\R^2)} \leq \| \theta_0 \|_{L^q(\R^2)}.
\end{equation}
\end{lemma}
\begin{proof}
Note that for smooth solutions after multiplying the nonlinear term $u\cdot \nabla \theta$ by $|\theta|^{q-1} sgn(\theta)$ and integrating the result over $\R^2$ the resulting term will vanish:
\begin{equation*}
\begin{split}
\int_{\R^2} &\left(\frac{\partial \theta}{\partial x_1} \frac{\partial}{\partial x_2} [(-\Delta)^{-\frac{1}{2}} \theta] - \frac{\partial \theta}{\partial x_2} \frac{\partial}{\partial x_1}[(-\Delta)^{-\frac{1}{2}} \theta] \right) |\theta|^{q-1} sgn(\theta) dx  \\
&= \frac{1}{q} \int_{\R^2} \left(\frac{\partial (|\theta|^q)}{\partial x_1} \frac{\partial}{\partial x_2} [(-\Delta)^{-\frac{1}{2}} \theta] - \frac{\partial (|\theta|^q)}{\partial x_2} \frac{\partial}{\partial x_1}[(-\Delta)^{-\frac{1}{2}} \theta] \right) dx = 0,
\end{split}
\end{equation*}
thanks to integration by parts. Consequently, multiplying \eqref{Q-gequation} by $|\theta|^{q-1} sgn(\theta)$, we obtain
\begin{equation*}
\int_{\R^2} \theta_t |\theta|^{q-1} sgn(\theta) dx + \kappa \int_{\R^2} (-\Delta)^\alpha \theta |\theta|^{q-1} sgn(\theta) dx = \int_{\R^2} f |\theta|^{q-1} sgn(\theta) dx.
\end{equation*}
Since for all the values $1 \leq q < \infty$ and $0 \leq \alpha \leq1$, as was shown in \cite{C-C} and \cite[Lemma 2.5]{C-C1},
\begin{equation*}
\int_{\R^2} (-\Delta)^\alpha \theta |\theta|^{q-1} sgn(\theta) dx \geq 0,
\end{equation*}
using  H\"older and Young inequalities, we get
\begin{equation*}
\frac{1}{q}\frac{d}{dt} \int_{\R^2} |\theta|^q dx  \leq  \int_{\R^2} f |\theta|^{q-1} sgn(\theta) dx
 \leq \frac{1}{q} \| f \|_{L^q(\R^2)}^q + \frac{q-1}{q} \| \theta \|_{L^q(\R^2)}^q.
\end{equation*}
Solving the above differential inequality  we get \eqref{maxp11} when $f \neq 0$ and \eqref{maxp} otherwise.
\end{proof}
Recall, see \cite[pp. 72-73]{C-D}, that to be able to extend globally in time the local $W^{2\alpha^-,p}(\R^2)$-solution constructed above, with bounded orbits of bounded sets, we need to have a subordination condition for it of the form:
\begin{equation*}
\| F(\theta(t, \theta_0)) \|_{L^p(\R^2)} \leq const(\|\theta(t, \theta_0)\|_Y) \left(1 + \| \theta(t, \theta_0)\|^\eta_{W^{2\alpha^-,p}(\R^2)}\right), \
t\in (0, \tau_{\theta_0}),
\end{equation*}
with certain auxiliary Banach space $W^{2\alpha,p}(\R^2)\subset Y$ and certain $\eta \in [0,1)$. The dependence of the solution on $\theta_0$ is marked in the notation here.

We will chose $Y: = L^{p}(\R^2) \cap L^{2p}(\R^2)$ for the problem \eqref{Q-gequation}.  By the H\"older inequality and the Nirenberg-Gagliardo estimate (e.g. \cite[p. 25]{C-D}) we obtain
\begin{equation*}
\begin{split}
\| F(\theta)\|_{L^p(\R^2)} &\leq c\| \theta \|_{L^{2p}(\R^2)} \| \theta \|_{W^{1, 2p}(\R^2)} + \| f \|_{L^p(\R^2)}+\kappa\| \theta\|_{L^p(\R^2)}   \\
&\leq const\left(\|\theta \|_{L^{p}(\R^2)}, \|\theta \|_{L^{2p}(\R^2)}\right) (1+\| \theta \|^\eta_{W^{2\alpha^-,p}(\R^2)})+ \| f \|_{L^p(\R^2)},
\end{split}
\end{equation*}
where $\frac{1}{2\alpha^--\frac{1}{p}}  < \eta<1$.
Since, due to Lemma \ref{MaxPrinc}, all the $L^{q}(\R^2)$, $q\geq p$, type norms are estimated a priori, we have the required bound sufficient for the global in time extendibility of the local solution. We thus claim:
\begin{theorem}
The local solution constructed in Theorem \ref{localexistence} will be extended globally in time. Moreover, the orbits of bounded subsets of $W^{2\alpha^-,p}(\R^2)$ are bounded. For arbitrary $\theta_0 \in W^{2\alpha^-,p}(\R^2)$ we have introduced a semigroup on $W^{2\alpha^-,p}(\R^2)$ by the formula:
\begin{equation*}
S(t)(\theta_0) = \theta(t),  \  t \geq 0,
\end{equation*}
where $\theta(t)$ denotes the corresponding to $\theta_0$  global in time $W^{2\alpha^-,p}(\R^2)$-solution of \eqref{Q-gequation}.
\end{theorem}

\subsection{Regularizing effect of the equation.} The solutions of semilinear equations with sectorial operator are regularised for $t>0 $. Our solution constructed above belongs to $W^{2\alpha,p}(\R^2)$ for $t>0$.

Multiplying equation \eqref{Q-gequation} by $t (-\Delta)^\alpha\theta$, we obtain
\begin{equation*}
\int_{\R^2} t \theta_t (-\Delta)^\alpha\theta dx + \int_{\R^2} t (-R_2\theta \frac{\partial\theta}{\partial x_1} + R_1\theta \frac{\partial\theta}{\partial x_2}) (-\Delta)^\alpha\theta dx + \kappa \int_{\R^2} t [(-\Delta)^\alpha\theta]^2 dx = \int_{\R^2} t f (-\Delta)^\alpha\theta dx.
\end{equation*}
The components are next estimated in the following way:
\begin{equation*}
\int_{\R^2} t \theta_t (-\Delta)^\alpha\theta dx = 
\frac{1}{2} \int_{\R^2} t \frac{\partial}{\partial t} [(-\Delta)^{\frac{\alpha}{2}} \theta]^2 dx   \\
= \frac{1}{2} \frac{d}{dt} \int_{\R^2} t [(-\Delta)^{\frac{\alpha}{2}} \theta]^2 dx  - \frac{1}{2} \int_{\R^2} [(-\Delta)^{\frac{\alpha}{2}} \theta]^2 dx,
\end{equation*}
where we need earlier estimate of $\theta$ in $L^2((0,T); H^\alpha(\R^2))$. The nonlinear term is estimated as
\begin{equation*}
t|\int_{\R^2} R_2\theta \frac{\partial\theta}{\partial x_1} (-\Delta)^\alpha\theta dx| \leq t \| R_2\theta \|_{L^{\infty^-}(\R^2)} \|\frac{\partial\theta}{\partial x_1}\|_{L^{2^+}(\R^2)} \| (-\Delta)^\alpha\theta]\|_{L^2(\R^2)},
\end{equation*}
where $\frac{1}{\infty^-} + \frac{1}{2^+} + \frac{1}{2}= 1$. Since $\| R_2\theta \|_{L^{\infty^-}(\R^2)} \leq const$ by Lemma \ref{MaxPrinc} and
\begin{equation*}
\|\frac{\partial\theta}{\partial x_1}\|_{L^{2^+}(\R^2)} \leq \| \theta \|_{W^{1,2^+}(\R^2)} \leq c \| \theta \|^{1-\nu}_{L^2(\R^2)} \|\theta\|^\nu_{H^{2\alpha}(\R^2)},
\end{equation*}
where $\nu \geq \frac{1 - \frac{1}{2^+}}{\alpha}$ will be chosen less than one, such nonlinear term will be controlled by $t \|\theta \|^2_{H^{2\alpha}(\R^2)}$.
Another components are estimated in a standard way.  

\section{  Solvability of subcritical \eqref{Q-gequation}, $\alpha \in (\frac{1}{2}, 1]$, in $H^{s+2\alpha^-}(\R^2)$.}\label{section4}
As in \cite{K-N-V}, it is also possible to chose $H^s(\R^2)$ as a {\it base space} where we consider the equation \eqref{Q-gequation}. In that case the resulting phase space will be equal $H^{2\alpha^- +s}(\R^2)$ and, when $s>1$, it will enjoy the embedding $H^{2\alpha^- +s}(\R^2) \subset W^{1,\infty}(\R^2)$.

{\bf Local solvability.} As in the previous case $X = L^p(\R^2)$ in subsection \ref{sub2.2}, we need to check that the nonlinearity \eqref{2.4} is Lipschitz on bounded sets as a map from $H^{2\alpha^- +s}(\R^2)$ into $H^s(\R^2)$, $s>1$. Indeed, as in \eqref{2.6}, we have
\begin{equation}\label{3.15}
\begin{split}
\| F(\theta_1) - F(\theta_2) \|_{H^s(\R^2)} &\leq \kappa \| \theta_1 -\theta_2 \|_{H^s(\R^2)}+\| R_2(\theta_1 -\theta_2) \frac{\partial \theta_1}{\partial x_1} + R_2 \theta_2 \frac{\partial(\theta_1-\theta_2)}{\partial x_1} \|_{H^s(\R^2)} \\
&+ \|  R_1(\theta_1-\theta_2) \frac{\partial \theta_1}{\partial x_2} + R_1 \theta_2 \frac{\partial(\theta_1-\theta_2)}{\partial x_2} \|_{H^s(\R^2)}.
\end{split}
\end{equation}
The second term above  we estimate using  the property (e.g. \cite[p. 115]{AD}), that {\it when $\Omega$ is a domain in $\R^N$ then $W^{m,r}(\Omega)$ is a Banach Algebra provided that $mr > N$}. In our case $2s >2$ since $s>1$. Hence
\begin{equation}\label{nondif}
\| R_2(\theta_1 - \theta_2) \frac{\partial \theta_1}{\partial x_1}\|_{H^s(\R^2)} \leq C \| R_2(\theta_1 - \theta_2) \|_{H^s(\R^2)} \| \frac{\partial \theta_1}{\partial x_1}\|_{H^s(\R^2)}\leq c \| |u_1 - u_2| \|_{H^s(\R^2)} \| \theta_1 \|_{H^{2\alpha^- +s}(\R^2)},
\end{equation}
where $u_i$, $i=1,2$, correspond to $\theta_i$ through relation \eqref{1.3}.\\
 Applying the  property \eqref{DU}  (see also  \cite[p. 12]{W1})
to the first term in \eqref{nondif} we get
\begin{equation*}
\| R_2(\theta_1 - \theta_2) \frac{\partial \theta_1}{\partial x_1}\|_{H^s(\R^2)} \leq c \| \theta_1 - \theta_2 \|_{H^{2\alpha^- +s}(\R^2)} \| \theta_1 \|_{H^{2\alpha^- +s}(\R^2)}.
\end{equation*}
The others components in \eqref{3.15} are estimated analogously. Consequently we obtain that
\begin{equation*}
\| F(\theta_1) - F(\theta_2) \|_{H^s(\R^2)} \leq c'(\| \theta_1 \|_{H^{2\alpha^- +s}(\R^2)}, \| \theta_2 \|_{H^{2\alpha^- +s}(\R^2)}) \| \theta_1-\theta_2\|_{H^{2\alpha^- +s}(\R^2)},
\end{equation*}
which proves local solvability of \eqref{Q-gequation} in $H^{2\alpha^- +s}(\R^2)$ phase space. More precisely, following \cite{HE, C-D}, we formulate:  

\begin{theorem}\label{sloc}
Let $s>1$ be fixed. Then, for $f\in H^s(\R^2)$ and for arbitrary $\theta_0 \in H^{2\alpha^- +s}(\R^2)$, there exists in the phase space $H^{2\alpha^- +s}(\R^2)$ a unique local in time mild solution $\theta(t)$ to the subcritical problem \eqref{Q-gequation}, $\alpha \in (\frac{1}{2},1]$. Moreover,  
\begin{equation*}
\theta \in C((0,\tau); H^{2\alpha+s}(\R^2)) \cap C([0,\tau); H^{2\alpha^- +s}(\R^2)), \  \theta_t \in C((0,\tau); H^{2\gamma+s}(\R^2)),
\end{equation*}
with arbitrary $\gamma < \alpha^-$. Here $\tau> 0$ is the 'live time' of that local in time solution. Moreover, the {\em Cauchy formula} is satisfied:
\begin{equation*}
\theta(t) = e^{-A_\alpha t} \theta_0 + \int_0^t e^{-A_\alpha(t-s)} F(\theta(s)) ds, \  t \in [0, \tau),
\end{equation*}
where $e^{-A_\alpha t}$ denotes the linear semigroup corresponding to the operator $A_\alpha := \kappa[(-\Delta)^\alpha +I]$ in $H^s(\R^2)$, and
$F$ is given by formula \eqref{2.4}.
\end{theorem}

{\bf Global solvability.} To guarantee the global in time solvability of \eqref{Q-gequation} in $H^{2\alpha^- +s}(\R^2)$ the a priori estimate \eqref{3.32} below will be used. It shows that $c'(\| \theta_1 \|_{H^{2\alpha^- +s}(\R^2)}, \| \theta_2 \|_{H^{2\alpha^- +s}(\R^2)})$ is bounded on the solutions. Consequently we have global Lipschitz continuity and boundedness of the nonlinear term $F$ as a map from $H^{2\alpha^- +s}(\R^2)$ to $H^s(\R^2)$.

\subsection{Further a priori estimates.}\label{FPE}
Following calculations in \cite[p. 1165]{W2} we are able to estimate higher Sobolev norms of the solutions to \eqref{Q-gequation}. Let $l \geq \alpha$ be fixed  and $f\in H^{l-\alpha}(\R^2)$. Multiplying the equation by $(-\Delta)^l\theta$ we obtain
\begin{equation}\label{3.22}
\frac{1}{2} \frac{d}{dt} \int_{\R^2} [(-\Delta)^{\frac{l}{2}} \theta]^2 dx + \kappa \int_{\R^2} [(-\Delta)^{\frac{l+\alpha}{2}}\theta]^2 dx = \int_{\R^2} (-\Delta)^{\frac{l-\alpha}{2}}f (-\Delta)^{\frac{l+\alpha}{2}}\theta dx - \int_{\R^2} u\cdot \nabla\theta (-\Delta)^l\theta dx.
\end{equation}
The nonlinear term is transformed as follows: 
\begin{equation}\label{3.24}
\begin{split}
|- \int_{\R^2} u\cdot \nabla\theta (-\Delta)^l\theta dx| &=  |\int_{\R^2} (-\Delta)^{\frac{l-\alpha}{2}}(u \cdot\nabla\theta) (-\Delta)^{\frac{l+\alpha}{2}}\theta dx| \\
&\leq \|(-\Delta)^{\frac{l+\alpha}{2}}\theta \|_{L^2(\R^2)} \|(-\Delta)^{\frac{l-\alpha}{2}} (-R_2\theta \frac{\partial\theta}{\partial x_1} + R_1\theta \frac{\partial\theta}{\partial x_2}) \|_{L^2(\R^2)}.
\end{split}
\end{equation}
Using a nontrivial estimate \eqref{commutator} with $\frac{1}{q} + \frac{1}{r} = \frac{1}{2}$ we obtain 
\begin{equation}\label{4.12ext}
\begin{split}
\|(-\Delta)^{\frac{l-\alpha}{2}} (-R_2\theta \frac{\partial\theta}{\partial x_1} &+ R_1\theta \frac{\partial\theta}{\partial x_2}) \|_{L^2(\R^2)}\\&\leq 
c\bigl[\|(-\Delta)^{\frac{l-\alpha}{2}} R_2\theta \|_{L^{q}(\R^2)} \|\frac{\partial\theta}{\partial x_1}\|_{L^{r}(\R^2)}   
+ \|R_2\theta\|_{L^{r}(\R^2)} \|(-\Delta)^{\frac{l-\alpha}{2}} \frac{\partial\theta}{\partial x_1}\|_{L^{q}(\R^2)}  
\\&+ \|(-\Delta)^{\frac{l-\alpha}{2}} R_1\theta \|_{L^{q}(\R^2)} \|\frac{\partial\theta}{\partial x_2}\|_{L^{r}(\R^2)} + \|R_1\theta\|_{L^{r}(\R^2)} \|(-\Delta)^{\frac{l-\alpha}{2}} \frac{\partial\theta}{\partial x_2}\|_{L^{q}(\R^2)}\bigr].
\end{split}
\end{equation}
Next we will use the property (e.g. \cite[p. 300]{M-S}) that for $f \in D((-\Delta)^\beta)$, $\beta>0$, $j=1,2$, 
$$
(-\Delta)^\beta R_j f = R_j(-\Delta)^\beta f,
$$
and  \eqref{Rieszobs}. Thus, estimate \eqref{4.12ext} extends to 
\begin{equation*}
\begin{split}
\|(-\Delta)^{\frac{l-\alpha}{2}} (-R_2\theta \frac{\partial\theta}{\partial x_1} + R_1\theta \frac{\partial\theta}{\partial x_2}) \|_{L^2(\R^2)} 
&\leq c  \bigl[\|(-\Delta)^{\frac{l-\alpha}{2}} \theta \|_{L^{q}(\R^2)} \|\theta \|_{W^{1,r}(\R^2)} + \|\theta\|_{L^{r}(\R^2)} \| \theta \|_{W^{l+1-\alpha,q}(\R^2)}\bigr].  
\end{split}
\end{equation*}  
 Taking  $q = \frac{1}{1-\alpha^-}$ and $r = \frac{2}{2\alpha^- -1}$,  by the Nirenberg-Gagliardo inequality,  we have
\begin{equation*}
\|\theta \|_{W^{l+1-\alpha,q}(\R^2)} \leq c \| \theta \|_{H^{l+\alpha}(\R^2)}^\eta\| \theta \|^{1-\eta}_{L^{2}(\R^2)},
\end{equation*}
\begin{equation*}
\|\theta \|_{W^{l-\alpha,q}(\R^2)} \leq c \|\theta \|^{\eta_1}_{H^{l+\alpha}(\R^2)} \|\theta \|^{1-\eta_1}_{L^{r}(\R^2)}, \quad \textrm{and}\quad \|\theta \|_{W^{1,r}(\R^2)} \leq c\|\theta \|^{\eta_2}_{H^{l+\alpha}(\R^2)} \|\theta \|^{1-\eta_2}_{L^{r}(\R^2)},
\end{equation*}
with some $\eta<1$ and $\eta_1+\eta_2 < 1$.
\\Using Young inequality, due to Lemma \ref{MaxPrinc}, the nonlinear term in \eqref{3.24} is estimated through
\begin{equation*}
|- \int_{\R^2} u\cdot \nabla\theta (-\Delta)^l\theta dx| \leq \frac{\kappa}{4}\| (-\Delta)^{\frac{l+\alpha}{2}}\theta \|^2_{L^2(\R^2)}+ const(\| \theta_0 \|_{L^{2}(\R^2)};\|\theta_0\|_{L^{r}(\R^2)}).
\end{equation*}
Consequently, from \eqref{3.22} and the above estimate we obtain a differential inequality:
\begin{equation}\label{3.32}
 \frac{d}{dt} \int_{\R^2} [(-\Delta)^{\frac{l}{2}} \theta]^2 dx + \kappa \int_{\R^2} [(-\Delta)^{\frac{l+\alpha}{2}}\theta]^2 dx \leq c(\| \theta_0\|_{L^\infty(\R^2)}, \| \theta_0\|_{L^2(\R^2)};\| f \|_{H^{l-\alpha}(\R^2)}),
\end{equation}
providing us, together with Lemma \ref{MaxPrinc}, the bound for $\|\theta \|_{H^l(\R^2)}$, $l\geq \alpha$, through $\| f \|_{H^{l-\alpha}(\R^2)}$.

\section{Asymptotic behavior of solutions to \eqref{lQ-gequation}.}
In this section we consider regularization of quasi-geostrophic  equation
\begin{equation*}
\theta_t  +u \cdot \nabla\theta = 0, \quad x \in \R^2,\, t>0,  \\
\end{equation*}
of the following form
\begin{equation}\label{lQ-gequation}
\begin{split}
&\theta_t + \lambda \theta +u \cdot \nabla\theta + \kappa(-\Delta)^\alpha \theta = f, \quad x \in \R^2,\, t>0,  \\
&\theta(0,x) = \theta_0(x),
\end{split}
\end{equation}
where $\kappa$, $\alpha$, and $u$ are like in \eqref{Q-gequation} and \eqref{1.2} and $\lambda>0$.
We changed the equation \eqref{Q-gequation} by adding to its left hand side the linear damping term $\lambda \theta$; compare \cite{GH} for discussion of the role of such a damping. This modification allow for existence of a global attractor. Using the techniques of the {\it tail estimates} (e.g. \cite{WA}), we study existence of such global attractor. Like for problem \eqref{Q-gequation} in section \ref{section4} we choose $H^s(\mathbb{R}^2)$, $s>1$, as a base space.  It is easy to check  that the global  in time solvability   and properties of solution of problem \eqref{Q-gequation}, obtained in  section \ref{section4}, are also true for solutions of \eqref{lQ-gequation}. The following theorem holds: 
 
\begin{theorem}\label{lsloc}
Let $s>1$ be fixed. Then, for $f\in H^{s}(\R^2)$ and for arbitrary $\theta_0 \in H^{2\alpha^- +s}(\R^2)$, there exists in the phase space $H^{2\alpha^- +s}(\R^2)$ a unique global in time mild solution $\theta(t)$ to the subcritical problem \eqref{lQ-gequation}, $\alpha \in (\frac{1}{2},1]$. Moreover,  
\begin{equation*}
\theta \in C((0,+\infty); H^{2\alpha+s}(\R^2)) \cap C([0,+\infty); H^{2\alpha^- +s}(\R^2)), \  \theta_t \in C((0,+\infty); H^{2\gamma+s}(\R^2)),
\end{equation*}
with arbitrary $\gamma < \alpha^-$.  Furthermore, the {\em Cauchy formula} is satisfied:
\begin{equation*}
\theta(t) = e^{-A_{\alpha,\lambda} t} \theta_0 + \int_0^t e^{-A_{\alpha,\lambda}(t-s)} F_1(\theta(s)) ds, \  t \in [0, +\infty),
\end{equation*}
where $e^{-A_{\alpha,\lambda} t}$ denotes the linear semigroup corresponding to the operator $A_{\alpha,\lambda} := \kappa(-\Delta)^\alpha +\lambda I$ in $H^s(\R^2)$, and
$$
F_1(\theta) = R_2 \theta \frac{\partial \theta}{\partial x_1} - R_1 \theta \frac{\partial \theta}{\partial x_2} + f.
$$
\end{theorem}
The global solution to \eqref{lQ-gequation} constructed in Theorem \ref{lsloc} define on the phase space $H^{s+2\alpha-}(\R^2)$ a semigroup
\begin{eqnarray*}
S_{\alpha,\lambda}(t)\theta_0=\theta(t,\theta_0),\quad t\geq 0,\,\alpha\in(\frac{1}{2},1],\,\lambda>0.
\end{eqnarray*}

\begin{remark}
For arbitrary $q \in [2,\infty)$, and a sufficiently regular solution of \eqref{lQ-gequation}, if $f$ is non-zero and in $L^q(\R^2)$, then
\begin{equation*} 
\|\theta(t,\cdot)\|^q_{L^q(\R^2)} \leq \| \theta_0 \|^q_{L^q(\R^2)} e^{-\frac{\lambda qt}{2}} + C(\lambda)\left(1-e^{-\frac{\lambda qt}{2}}\right) \|f\|_{L^q(\R^2)}^q.
\end{equation*}
\end{remark}
\begin{remark} \label{HS}
The estimates of higher Sobolev norms $H^l(\mathbb{R}^2)$, $l>\alpha$, of solutions to \eqref{Q-gequation} obtained in subsection \ref{FPE} hold also for solutions to \eqref{lQ-gequation}.
\end{remark}

We  study now existence of the global attractor for the semigroups generated by  the problem \eqref{lQ-gequation}.
We choose the smooth (at least $C^2$, but we prefer $\eta \in C^\infty$) cut-off function $\eta: \mathbb{R}^2 \rightarrow [0,1]$,
\begin{equation*}
\eta(x)=
\begin{cases}
1,\quad |x|\geqslant 2,\\
0,\quad |x|\leqslant 1.
\end{cases}
\end{equation*} 

\begin{lemma}\label{Lemma5}
For any $\alpha \in (0,1)$, there exists a constant $M=M(\alpha,\eta)>0$ such that
\[
|(-\Delta)^{\alpha}\eta(x)|\leqslant M <\infty \quad \text{for all}~x\in \mathbb{R}^2.
\]
\end{lemma}
Let $\eta_k(\cdot)=\eta(\frac{\cdot}{k})$, $k=1,2,\ldots$. Then the
following identity is obvious
\begin{lemma}\label{Lemma51}
For any $s\in (0,2)$,
\[
(-\Delta)^{\frac{s}{2}}\eta_k(x)=\frac{1}{k^s}(-\Delta)^{\frac{s}{2}} 
\eta(z)|_{z=\frac{x}{k}}.
\]
\end{lemma} 

We  obtain now for solutions $\theta$ of \eqref{lQ-gequation} the, so called, {\it{tail estimates}} in
$L^2(\mathbb{R}^2)$ as introduced in \cite{WA}:
\begin{lemma}\label{TE}
For each $\epsilon > 0$ and
arbitrary $\theta_0 \in H^{2\alpha^-+s}(\mathbb{R}^2)$, there exist $k = k(\epsilon;
\|\theta_0\|_{L^2(\mathbb{R}^2)};\|\theta_0\|_{L^4(\mathbb{R}^2)})$ and $T = T(\epsilon;
\|\theta_0\|_{L^2(\mathbb{R}^2)};\|\theta_0\|_{L^4(\mathbb{R}^2)})$
such that the corresponding to $\theta_0$ solution $\theta(t)$ of \eqref{Q-gequation} satisfies
\begin{equation*}
\int_{O_k}|\theta(t)|^2 dx \leq \epsilon\quad \textrm{for all } t\geq T,
\end{equation*}
where $\mathcal{O}_k=\{x\in\mathbb{R}^2\colon |x|\geq k \}$.
\end{lemma}
\begin{proof}
Taking the scalar product of \eqref{lQ-gequation} with $\theta \eta_k$, we obtain
 \begin{equation*}
\int_{\R^2} \theta_t\theta \eta_k dx +\lambda\int_{\R^2} \theta^2 \eta_k dx =-\int_{\R^2} u\cdot \nabla\theta \theta \eta_k dx- \kappa \int_{\R^2} (-\Delta)^\alpha \theta \theta \eta_k dx +\int_{\R^2} f \theta \eta_kdx.
\end{equation*}
We will transform the components one by one. We have
\begin{equation*}
\int_{\R^2} \theta_t\theta \eta_k dx =\frac{1}{2} \frac{d}{dt}
\int_{\mathbb{R}^N} \theta^2 \eta_k dx.
\end{equation*}
Integrating by parts, due to \eqref{Rieszobs}, we get
\begin{eqnarray*}
\begin{split}
-&\int_{\R^2} u\cdot \nabla\theta \theta \eta_k dx=-\frac{1}{2}\int_{\R^2} u\cdot \nabla\theta^2 \eta_k dx=\frac{1}{2}\int_{\R^2}\bigl(\frac{\partial}{\partial x_1}\theta^2\frac{\partial}{\partial x_2}\psi\eta_k -\frac{\partial}{\partial x_2}\theta^2\frac{\partial}{\partial x_1}\psi\eta_k\,\bigr) dx\\&
=\frac{1}{2}\int_{\R^2} \bigl(-\theta^2\frac{\partial}{\partial x_1}\frac{\partial}
{\partial x_2}\psi\eta_k-\theta^2\frac{\partial}{\partial x_2}\psi\frac{\partial}{\partial x_1}\eta_k
+\theta^2\frac{\partial}{\partial x_2}\frac{\partial}
{\partial x_1}\psi\eta_k+\theta^2\frac{\partial}{\partial x_1}\psi\frac{\partial}{\partial x_2}\eta_k\,\bigr) dx
\\&\leq\frac{1}{2}\int_{\R^2}\theta^2|u| |\nabla\eta_k| dx \leq \frac{c}{k}\left(\int_{\R^2}\theta^4 dx+\int_{\R^2}|u|^2  dx\right)
\leq \frac{const}{k}\left(\int_{\R^2}\theta^4 dx+\int_{\R^2}\theta^2  dx\right).
\end{split}
\end{eqnarray*}
Thanks  to the pointwise estimate from \cite{C-C1} which states that
\begin{equation*}
\forall_{\alpha \in [0,1]} \forall_{\phi \in C^2_0(\R^N)}  \  2\phi (-\Delta)^\alpha \phi \geq (-\Delta)^\alpha(\phi^2),
\end{equation*}
 we have
 \begin{equation*}
 -\kappa \int_{\R^2} (-\Delta)^\alpha \theta \theta \eta_k dx \leq -\frac{\kappa}{2} \int_{\mathbb{R}^2}
\theta^2 (-\Delta)^\alpha(\eta_k) dx\leq \frac{const}{k^{2\alpha}}\int_{\mathbb{R}^2}
\theta^2  dx.
\end{equation*}
The last component  is transformed as follows
\begin{eqnarray*}
\int_{\R^2} f \theta \eta_k dx\leq \frac{1}{2\lambda} \int_{\R^2} f^2\eta_k dx+\frac{\lambda}{2}\int_{\R^2} \theta^2\eta_k dx\leq \frac{1}{2\lambda} \int_{\mathcal{O}_k} f^2 dx+\frac{\lambda}{2}\int_{\R^2} \theta^2\eta_k dx.
\end{eqnarray*}
Consequently, 
\begin{equation}\label{aa}
\frac{1}{2} \frac{d}{dt}
\int_{\mathbb{R}^N} \theta^2 \eta_k dx +\frac{\lambda}{2}
\int_{\mathbb{R}^N} \theta^2 \eta_k dx \leq\frac{const}{k^{2\alpha}}\int_{\mathbb{R}^2}\theta^2  dx+\frac{const}{k}\left(\int_{\R^2}\theta^4 dx+\int_{\R^2}\theta^2  dx \right)+\frac{1}{2\lambda}\int_{\mathcal{O}_k} f^2 dx.
\end{equation}
Note that, since $f\in  L^2(\R)$, then
\begin{equation}\label{aaa}
\int_{\mathcal{O}_k} f^2 dx\to 0\quad \textrm{as} \quad k\to \infty.
\end{equation}
Combining \eqref{aa}  together with \eqref{aaa}, and $L^q(\R^2)$ uniform in time estimate of solution $\theta$, $q\in[2,\infty)$,  we obtain thesis.
 \end{proof}
\begin{lemma}\label{ACH}
The semigroups $\{S_{\alpha,\lambda}(t)\}_{t\geq 0}$  are asymptotically compact in $H^{s+2\alpha^-}(\R^2)$.
\end{lemma}
\begin{proof} The asymptotic compactness in $H^{s+2\alpha^-}(\R^2)$ follows directly from Lemma \ref{TE}, Remark \ref{HS} and interpolation inequality.
\end{proof}
\begin{remark}
As a consequence of the abstract results in \cite{TE, Ra}, asymptotic compactness 
reported in Lemma \ref{ACH} and estimates in  Remark \ref{HS}, the semigroups $\{S_{\alpha,\lambda}(t)\}_{t\geq 0}$  possess $H^{s+2\alpha^-}(\R^2)$ global attractors $\mathcal{A}_{\alpha,\lambda}$.
\end{remark} 

\section{Critical equation \eqref{Q-gequation}; $\alpha= \frac{1}{2}$.}\label{critical}
\subsection{Passing to the limit.} A precise description of letting $\alpha \to \frac{1}{2}^+$ in the equation \eqref{Q-gequation} is given next. Technical details used in the main considerations below are discussed in the following subsection. 

In this section we consider the solutions $\theta^\alpha$ (we add the superscript for clarity) constructed in section \ref{section4} on the base space $H^s(\R^2), s >1$. Such solutions, for any $\alpha \in (\frac{1}{2},1]$, are varying continuously in $H^{1+s}(\R^2)$, $s>1$, hence also in each of the spaces $W^{1,p}(\R^2)$, $2 \leq p \leq +\infty$. 

In particular they fulfill the {\it uniform in $\alpha \in (\frac{1}{2}, 1]$} estimates in $L^p(\R^2)$ of Lemma \ref{MaxPrinc}, the {\it main information} allowing us to let $\alpha \to \frac{1}{2}^+$ in the equation \eqref{Q-gequation}. More precisely, for such solutions, we have:
\begin{equation}
\exists_{const >0} \forall_{\alpha \in (\frac{1}{2}, \frac{3}{4}]} \forall_{p \in [2,+\infty)}  \  \|\theta^\alpha \|_{L^\infty([0,T]; L^p(\R^2))} \leq const,
\end{equation}
where $T>0$ is fixed arbitrary large. To work with {\it sectorial positive} operator $A_\alpha := \kappa[(-\Delta)^\alpha + I]$, we add to both sides of \eqref{Q-gequation} the term $\kappa \theta$ to obtain:
\begin{equation}\label{6.1}
\begin{split}
&\theta^\alpha_t + u \cdot \nabla\theta^\alpha + A_\alpha \theta^\alpha = f + \kappa\theta^\alpha, \ x \in \R^2, t>0,  \\
&\theta^\alpha(0,x) = \theta_0(x),
\end{split}
\end{equation}
 We look at \eqref{6.1} as an equation in $L^p(\R^2)$, $p\geq 2$, and 'multiply' it by the test function $(A_{\alpha}^{-1})^*\phi$ where $\phi \in \mathcal{L}_{-2}^q(\R^2), \frac{1}{p}+ \frac{1}{q}=1$, to get:
\begin{equation}\label{6.22}
<[\theta^\alpha_t + u \cdot \nabla\theta^\alpha], (A_\alpha^{-1})^*\phi>_{L^p, L^q}  = -<A_\alpha\theta^\alpha, (A_\alpha^{-1})^*\phi>_{L^p, L^q} + <[f + \kappa \theta^\alpha], (A_\alpha^{-1})^*\phi>_{L^p, L^q}. 
\end{equation} 
We will discuss now the convergence of the terms in \eqref{6.22} one by one. Note that when $\alpha \to \frac{1}{2}^+$, then by Lemma \ref{lem6.2}, $(A_\alpha^{-1})^*\phi \to (A_{\frac{1}{2}}^{-1})^*\phi$ for $\phi \in \mathcal{L}^q_{-2}(\R^2)$.  Thanks to boundedness of $\theta^\alpha$ in $L^p(\R^2)$, $p \geq 2$, uniform in $\alpha \in (\frac{1}{2}, 1]$ (Lemma \ref{MaxPrinc}), adding and subtracting, we obtain:   
\begin{equation}\label{mic}
\begin{split}
<[f + \kappa \theta^\alpha], (A_\alpha^{-1})^*\phi>_{L^p, L^q} = &<[f + \kappa \theta^\alpha], (A_\alpha^{-1})^*\phi -(A_{\frac{1}{2}}^{-1})^*\phi>_{L^p, L^q} + <[f + \kappa \theta^\alpha], (A_{\frac{1}{2}}^{-1})^*\phi>_{L^p, L^q}  \\  
&\to <[f + \kappa \theta], (A_{\frac{1}{2}}^{-1})^*\phi>_{L^p, L^q},
\end{split}
\end{equation}
where $\theta$ is the weak limit of $\theta^\alpha$ in $L^p(\R^2)$ as $\alpha\to \frac{1}{2}^+$ (over a sequence $\{\alpha_n\}$ convergent to $\frac{1}{2}^+$; various sequences may lead to various weak limits). \newline 
For the intermediate term, we have the equality:
\begin{equation}
<A_\alpha\theta^\alpha, (A_\alpha^{-1})^*\phi>_{L^p, L^q} = <\theta^\alpha, \phi>_{L^p, L^q}.
\end{equation}  
Moreover, since $\theta^\alpha \in W^{1,p}(\R^2)$, $p \geq 2$, we have: $\theta^\alpha \in L^p(\R^2)$ and $A_{\frac{1}{2}}\theta^\alpha=: T^\alpha \in L^p(\R^2)$, which gives that $\theta^\alpha = A_{\frac{1}{2}}^{-1}T^\alpha$ for some $T^\alpha \in  L^p(\R^2)$. Inserting to the equation above we obtain:  
\begin{equation*}
<\theta^\alpha, \phi>_{L^p, L^q}  = <A_{\frac{1}{2}}^{-1}T^\alpha, \phi>_{L^p, L^q} = <T^\alpha, (A_{\frac{1}{2}}^{-1})^*\phi>_{L^p, L^q}.
\end{equation*} 
Since the left hand side has a limit as $\alpha_n \to \frac{1}{2}^+$, the right hand side also has the same limit:
\begin{equation}\label{Tdef}
<T^\alpha, (A_{\frac{1}{2}}^{-1})^*\phi>_{L^p, L^q}  \to <\Theta, (A_{\frac{1}{2}}^{-1})^*\phi>_{L^p, L^q} = <\theta, \phi>_{L^p, L^q} \  \text{for all} \  \phi \in \mathcal{L}^q_{-2}(\R^2);  
\end{equation} 
note that $\mathcal{L}^q_{-2}(\R^2)$ is dense in $L^q(\R^2)$ and that $\theta, \Theta \in L^p(\R^2)$.

Consequently, returning to \eqref{6.22}, we see that the first term also has a limit:
\begin{equation}\label{added}
<[\theta^\alpha_t + u \cdot \nabla\theta^\alpha], (A_\alpha^{-1})^*\phi>_{L^p, L^q} \to \omega_\phi, 
\end{equation}
as $\alpha \to \frac{1}{2}^+$.  \newline 
Considering the limits together, taken over the same sequence $\{\alpha_n\}$ convergent to $\frac{1}{2}^+$, we find a weak form of the limit equation:   
\begin{equation} 
\forall_{\phi\in \mathcal{L}^q_{-2}(\R^2)}  \   \omega_\phi = -<\theta, \phi>_{L^p, L^q} + <f + \kappa \theta, (A_{\frac{1}{2}}^{-1})^*\phi>_{L^p, L^q}.
\end{equation} 
Remembering that the set $\{(A_{\frac{1}{2}}^{-1})^*\phi; \phi \in \mathcal{L}^q_{-2}(\R^2) \}$ is dense in $L^q(\R^2)$, the right hand side above defines a unique element in $L^p(\R^2)$.    

{\bf Separation of terms.} The time derivative will be separated from the term $[\theta^\alpha_t+ u \cdot \nabla \theta^\alpha]$ when letting $\alpha \to \frac{1}{2}^+$. More precisely we have: 
\begin{remark}\label{explanation}
Since the approximating solutions $\theta^\alpha$ satisfy
\begin{equation}
\theta^\alpha \in L^\infty(0,T;L^p(\R^2)), \theta_t^\alpha \in L^2(0,T;L^p(\R^2)), \    \alpha \in (\frac{1}{2},\frac{3}{4}], 
\end{equation}
then by \cite[Lemma 1.1, Chapt.III]{TE1} 
\begin{equation}
\forall_{\eta \in L^q(\R^2)} <\theta_t^\alpha, \eta>_{L^p, L^q} = \frac{d}{dt}<\theta^\alpha,\eta>_{L^p, L^q} \to \frac{d}{dt}<\theta,\eta>_{L^p, L^q},
\end{equation}
(here $\frac{1}{p} + \frac{1}{q} =1$)  the derivative $\frac{d}{dt}$ and the convergence are in $\mathcal{D}'(0,T)$ (space of the 'scalar distributions'). Consequently, 
\begin{equation}
\omega_\phi = \frac{d}{dt} <\theta, (A_{\frac{1}{2}}^{-1})^*\phi>_{L^p, L^q} + \omega^1_\phi,
\end{equation}
where $\omega^1_\phi$ is a limit in $\mathcal{D}'(0,T)$ of $<u\cdot \nabla\theta^\alpha, (A_\alpha^{-1})^*\phi>_{L^p, L^q}$ over a chosen sequence $\alpha_n \to \frac{1}{2}^+$. 
\end{remark}

The construction presented above allow us to formulate the following theorem:  
\begin{theorem}\label{critth}
Let $\{ \theta^\alpha \}_{\alpha\in (\frac{1}{2},\frac{3}{4}]}$ be the set of regular $H^{s+2\alpha^-}(\R^2)$ solutions to subcritical equation \eqref{Q-gequation}. Such solutions are, in particular, bounded in each space $L^p(\R^2)$ for $p \in [2,+\infty)$, uniformly in $\alpha$. As a consequence of that and the smoothness properties of regular  solutions (they vary continuously in $W^{1,p}(\R^2)$), for arbitrary sequence $\{ \alpha_n \} \subset (\frac{1}{2},\frac{3}{4}]$ convergent to $\frac{1}{2}^+$ we can find a subsequence $\{ \alpha_{n_k} \}$ that the corresponding sequence $\{ \theta^{n_k} \}$ converges weakly in $L^p(\R^2)$ to a function $\theta$ fulfilling the equation:  
\begin{equation}\label{6.12}  
\forall_{\phi\in \mathcal{L}^q_{-2}(\R^2)}  \   \omega_\phi = -<\theta, \phi>_{L^p, L^q} + <f + \kappa \theta, (A_{\frac{1}{2}}^{-1})^*\phi>_{L^p, L^q}.
\end{equation}  
Due to denseness of the set $\mathcal{L}^q_{-2}(\R^2)$ (see Definition \ref{mathcalL}) in $L^q(\R^2), \frac{1}{p}+\frac{1}{q} =1$, the right hand side of \eqref{6.12} defines a unique element in $L^p(\R^2)$. The left hand side $\omega_\phi$ is defined in \eqref{added} and discussed in Remark \ref{explanation}. We will call such $\theta$ a {\em weak $L^p$ solution} to the critical equation \eqref{Q-gequation}, $\alpha=\frac{1}{2}$.  
\end{theorem} 

\subsection{Some technicalities.} When passing to the limit in the considerations above it was important that the estimates, in particular the constants in it, can be taken {\em uniform} in $\alpha$. Therefore, in the technical lemmas below we need to care on a very precise expression of that uniformity. Even some estimates can be found in the literature, usually such uniformity is not clear from the presentation, thus we include here the proofs for completeness.

\begin{lemma}\label{negativeconv}
When $(\frac{1}{2},1] \ni \alpha \to \frac{1}{2}^+$, then $A^{-\alpha} \to A^{-\frac{1}{2}}$ pointwise on the domain of $A^{-1}$, where $A$ is a sectorial non-negative operator in a reflexive Banach space $X$.
\end{lemma}
\begin{proof}
Indeed, using the formula (e.g. \cite[p. 175]{M-S}) for negative powers of non-negative operators:
\begin{equation*}
A^{-\alpha} \phi = \frac{\sin(\pi\alpha)}{\pi} \int_0^\infty \lambda^{-\alpha} (A+\lambda I)^{-1} \phi d\lambda,
\end{equation*}
we find that, for  $\phi \in D(A^{-1})$ and any $L >1$,  
\begin{equation}\label{5.7}
\| A^{-\alpha} \phi - A^{-\frac{1}{2}}\phi \|_X = \frac{1}{\pi} \|\int_0^\infty \bigl[\frac{\sin(\pi\alpha)}{\lambda^\alpha} - \frac{1}{\lambda^{\frac{1}{2}}}\bigr] (A+\lambda I)^{-1} \phi d\lambda\|_X \leq \frac{1}{\pi} \|\int_0^L + \int_L^\infty \|_X.
\end{equation}
Further we will need the following two estimates:
\begin{equation}\label{dwa}
\begin{split}
&\|(A+\lambda I)^{-1} \|_X \leq \frac{M}{|\lambda|},  \\
&\|(A+\lambda I)^{-1} \phi\|_X = \|A(A+\lambda I)^{-1} A^{-1}\phi\|_X = \|[I - \lambda (A+\lambda I)^{-1}] A^{-1}\phi\|_X \leq  (M+1)\|A^{-1} \phi\|_X.
\end{split}
\end{equation}
Consequently we will estimate the right hand side terms in  \eqref{5.7} as follows:
\begin{equation*}
\frac{1}{\pi} \int_0^L \bigl|\frac{\sin(\pi\alpha)}{\lambda^\alpha} - \frac{1}{\lambda^{\frac{1}{2}}}\bigr| \|(A+\lambda I)^{-1} \phi\|_X d\lambda
\leq \frac{M+1}{\pi} \|A^{-1} \phi\|_X \int_0^L \lambda^{-\alpha} \bigl|\sin(\pi\alpha) - \lambda^{\alpha -\frac{1}{2}} \bigr| d\lambda,
\end{equation*}
\begin{equation*}
\frac{1}{\pi} \int_L^\infty \bigl|\frac{\sin(\pi\alpha)}{\lambda^\alpha} - \frac{1}{\lambda^{\frac{1}{2}}}\bigr| \|(A+\lambda I)^{-1} \phi\|_X d\lambda
\leq \frac{M\|\phi \|_X}{\pi} \int_L^\infty \left(\frac{1}{\lambda^{\alpha+1}} +\frac{1}{\lambda^{\frac{3}{2}}} \right) d\lambda  
\leq \frac{4M\| \phi\|_X}{L^{\frac{1}{2}}\pi}. 
\end{equation*}
We will chose now $L>1$ so large that the second term is less than $\epsilon\| \phi\|_X$, next chose $\alpha > \frac{1}{2}$ so close to $\frac{1}{2}$ that the first term is less than $\epsilon \|A^{-1}\phi\|_X$. Consequently,
\begin{equation}\label{6.16}
\| A^{-\alpha} \phi - A^{-\frac{1}{2}}\phi \|_X \leq \epsilon \left(\| \phi\|_X + \|A^{-1}\phi\|_X\right).
\end{equation}
This proves $A^{-\alpha} \to A^{-\frac{1}{2}}$ pointwise in the domain of $A^{-1}$. Indeed,  recall here the general definition of the domain of linear operator in the space $X$ applied to $A^{-1}$:
\begin{equation*}
D(A^{-1}) = \{\phi \in X; A^{-1}\phi \in X \}.
\end{equation*}
Thus, the convergence in  \eqref{6.16} holds for $\phi \in D(A^{-1})$.
\end{proof}
In fact, when passing to the limit in \eqref{6.1}, we need a variant of the above lemma which uses the formula (e.g. \cite[p. 117]{M-S}):
\begin{equation*}
(I + A^\alpha)^{-1} = \frac{\sin(\alpha\pi)}{\pi} \int_0^\infty \frac{\lambda^\alpha}{\lambda^{2\alpha}+ 2\lambda^\alpha cos(\pi\alpha)+1} (\lambda I+A)^{-1} d\lambda,
\end{equation*}
valid for $\alpha \in \mathbb{C}: \Re \alpha > |\alpha|^2$; in particular for $0<\alpha<1$. Applying the above expression to $A = (-\Delta)$, with $\alpha$ and with $\frac{1}{2}$, we have:
\begin{lemma}\label{lem6.2} 
Let $p\in[2,+\infty)$ be arbitrary. When $(\frac{1}{2},1] \ni \alpha \to \frac{1}{2}^+$, then $A_\alpha^{-1} := \frac{1}{\kappa}[(-\Delta)^\alpha + I]^{-1} \to A_{\frac{1}{2}}^{-1}$ pointwise in 
\begin{equation*}
D((-\Delta)^{-1}) = \{\phi \in L^p(\R^2); (-\Delta)^{-1}\phi \in L^p(\R^2)\} = \mathcal{L}^p_{-2}(\R^2);  
\end{equation*}
compare \eqref{elpe} for the definition of $\mathcal{L}^p_{-2}(\R^2)$. The same proof works for a general sectorial non-negative operator in a reflexive Banach space $X$.
\end{lemma}
\begin{proof}
We need to estimate the difference:
\begin{equation*}
[I+(-\Delta)^\alpha]^{-1} - [I+(-\Delta)^{\frac{1}{2}}]^{-1} = \frac{1}{\pi} \int_0^\infty \bigl[ \frac{\sin(\pi\alpha) \lambda^\alpha}{\lambda^{2\alpha}+ 2\lambda^\alpha \cos(\pi\alpha)+1} - \frac{\sin(\frac{\pi}{2}) \lambda^\frac{1}{2}}{\lambda+ 2\lambda^\frac{1}{2} \cos(\frac{\pi}{2})+1} \bigr] (\lambda I + (-\Delta))^{-1} d\lambda.
\end{equation*}
Now the proof goes as in the previous lemma separating the integral into two parts; over $(0,L)$ and $(L,\infty)$, and estimating them using \eqref{dwa}. More precisely, for $\phi \in D((-\Delta)^{-1})$, we have:
\begin{equation*}
\begin{split}
\frac{1}{\pi} \int_0^L \bigl| &\frac{\sin(\pi\alpha) \lambda^\alpha}{\lambda^{2\alpha} + 2\lambda^\alpha \cos(\pi\alpha)+1} - \frac{\lambda^\frac{1}{2}}{\lambda+1} \bigr| \|((-\Delta)+\lambda I)^{-1} \phi\|_{L^p(\R^2)} d\lambda     \\
&\leq \frac{M+1}{\pi} \|(-\Delta)^{-1} \phi\|_{L^p(\R^2)} \int_0^L \bigl| \frac{\sin(\pi\alpha) \lambda^\alpha}{\lambda^{2\alpha}+ 2\lambda^\alpha \cos(\pi\alpha)+1} - \frac{ \lambda^\frac{1}{2}}{\lambda+ 1} \bigr| d\lambda,
\end{split}
\end{equation*}
\begin{equation*}
\begin{split}
\frac{1}{\pi} \int_L^\infty &\bigl| \frac{\sin(\pi\alpha) \lambda^\alpha}{\lambda^{2\alpha}+ 2\lambda^\alpha \cos(\pi\alpha)+1} - \frac{\lambda^\frac{1}{2}}{\lambda+1} \bigr| \|((-\Delta)+\lambda I)^{-1} \phi \|_{L^p(\R^2)} d\lambda \\
&\leq \frac{M\| \phi\|_{L^p(\R^2)}}{\pi} \int_L^\infty  \left( \frac{ \lambda^{-1-\alpha}}{1 + 2\lambda^{-\alpha} \cos(\pi\alpha)+ \lambda^{-2\alpha}} + \frac{ \lambda^{-\frac{1}{2}}}{\lambda+1}  \right) d\lambda.
\end{split}
\end{equation*}
Note that the two terms under the final integral are bounded, moreover uniformly in $\alpha \in (\frac{1}{2},\frac{3}{4}]$. For such a range of parameter we have:
\begin{equation*}
 \frac{ \lambda^{-1-\alpha}}{1 + 2\lambda^{-\alpha} \cos(\pi\alpha)+ \lambda^{-2\alpha}} \leq
\frac{\lambda^{-1-\alpha}}{1 - \cos^2(\frac{3\pi}{4})}= \frac{2}{\lambda^{1+\alpha}}.
\end{equation*}
We will chose next $L>1$ so large that the integral over $(L,\infty)$ is less than $\kappa\epsilon\| \phi\|_{L^p(\R^2)}$, next chose $\alpha > \frac{1}{2}$ so close to $\frac{1}{2}$ that the integral over $(0,L)$ is less than $\kappa\epsilon \|(-\Delta)^{-1}\phi\|_{L^p(\R^2)}$. Consequently,
\begin{equation*}
\frac{1}{k}\| [I+(-\Delta)^\alpha]^{-1}\phi - [I+(-\Delta)^{\frac{1}{2}}]^{-1}\phi \|_{L^p(\R^2)} \leq \epsilon \left(\| \phi\|_{L^p(\R^2)} + \|(-\Delta)^{-1}\phi\|_{L^p(\R^2)}\right).
\end{equation*}
This proves $A_\alpha^{-1}\to A_{\frac{1}{2}}^{-1}$ pointwise in $\mathcal{L}^p_{-2}(\R^2)$.      
\end{proof}

When passing to the limit constructing weak solution to the critical equation \eqref{6.1}, we were using the 'test functions' $\phi \in \mathcal{L}^q_{-2}(\R^2)$. In that calculations we need to be sure that for such $\phi$, $(A_\alpha^{-1})^* \phi$ and $(A_{\frac{1}{2}}^{-1})^*\phi$ are in $L^q(\R^2)$ 
(rather; form a dense subset). 

This property follows from the observation taken from \cite[Corollary 12.2.5 (iii)]{M-S}, or \cite[Corollary 4.7]{M-S-P}, which states that:
\begin{corollary}
If $1<q<+\infty$ and $0< \nu < \frac{N}{2q}$, then $[R_\nu]_q = (-\Delta_q)^{-\nu}$ is a sectorial operator of amplitude $\frac{\pi}{2}$.   
Moreover, if $\nu, \beta \in \mathbb{C}$ are such that $0< \Re\nu< \Re\beta < \frac{N}{2q}$, then
\begin{equation}
D([R_\beta]_q) \subset D([R_\nu]_q).
\end{equation}

We immediately obtain from the above that: $\mathcal{L}^q_{-2}(\R^2) \subset \mathcal{L}^q_{-2\alpha}(\R^2) \subset \mathcal{L}^q_{-1}(\R^2)$ when $(1,2] \ni q$ (note, that we consider $\alpha \in (\frac{1}{2}, \frac{3}{4}]$ and $p \geq 2$, and $\frac{1}{p}+ \frac{1}{q} =1$ so that $q$ near 1 are interesting). 
\end{corollary} 

The information reported in Lemma \ref{lem6.2} will be used for letting $\alpha \to \frac{1}{2}^+$ in the equation \eqref{Q-gequation} provided the test functions will vary in $D((-\Delta)^{-1})$. We recall a characterization of this domain. We start with the definition of the spaces $\mathcal{L}^p_{\gamma}(\R^N)$ (e.g. \cite[p. 1164]{W2}); note that the definition in \cite[Chapt. V]{S}, using the powers of $(I-\Delta)^{-\frac{1}{2}}$ instead of $(-\Delta)^{-\frac{1}{2}}$, leads to usual fractional order Sobolev spaces.  
\begin{definition}\label{mathcalL}
For $1\leq p \leq \infty$ and $\gamma \in \R$ we define
\begin{equation}\label{elpe}
\mathcal{L}^p_{\gamma}(\R^N) = \{f \in L^p(\R^N)\colon f = I^\gamma g := (-\Delta)^{-\frac{\gamma}{2}}g \  \text{for certain} \  g \in L^p(\R^N) \},
\end{equation}
normed by $\| f \|_{\mathcal{L}^p_{\gamma}(\R^N)} = \| g \|_{L^p(\R^N)}$.
\end{definition}
It is known (e.g. \cite[p. 303]{M-S}), that when $0 <\gamma <m, 1 <p < \infty$, then:
\begin{equation*}
\mathcal{L}^p_\gamma(\R^2) \doteq  \mathcal{S}^{\gamma,p}(\R^N) \doteq [L^p(\R^N), W^{m,p}(\R^N)]_{\frac{\gamma}{m}},
\end{equation*}
the last bracket stands for the complex interpolation.  

Following \cite[\S 25]{S-K-M} we recall here the description of the space $D((-\Delta)^{-1}) = I^2(L^p(\R^N)) \cap L^p(\R^N)$ (we are using it in Lemma \ref{lem6.2}).   
\begin{observation}\label{Lizorkin}
Characterization of the spaces $L^{\gamma}_{p,r}(\R^N)$ for arbitrary $\gamma >0$, $1\leq p < \infty$, $1\leq r <\infty$, was given in \cite{S-K-M} \S 25, $7^o$. The auxiliary spaces there are the {\em Lizorkin spaces}
\begin{equation*}
\Phi = \{\phi \in  \mathscr{S}(\R^N)\colon (D^j\phi)(0) =0,\, |j| = 0,1,2,...\}, 
\end{equation*}
which contain in particular $\phi(x) = \exp(-|x|^2 -|x|^{-2})$. Next, the class $\Psi$, consisting of the Fourier transforms of all elements of $\Phi$ is introduced:  
\begin{equation*}
\Psi = \mathcal{F}(\Phi) = \{\psi \in  \mathscr{S}(\R^N)\colon \psi = \mathcal{F}(\phi),\, \phi\in \Phi\}.
\end{equation*}

The last class is characterized as all the {\it Schwartz functions} $\psi$, that are orthogonal to the polynomials:
\begin{equation*}
\int_{\R^N} x^j \psi(x) dx = 0, \quad  |j|=0,1,2,... .
\end{equation*}
Denoting 
\begin{equation*}
I^\gamma(L^p(\R^N)) = \{f\colon f= I^\gamma(\psi),\, \psi \in L^p(\R^N)\},
\end{equation*}
the class of the {\it Riesz potentials} of functions in $L^p(\R^N)$, when $\gamma \geq N$, the space $I^\gamma(L^p(\R^N))$ is understood in a sense of distributions on the Lizorkin class $\Psi$:
\begin{equation*}
<I^\gamma \psi,\omega> = <\psi, I^\gamma\omega>, \quad  \omega \in \Psi.
\end{equation*} 
\end{observation}
Description of the class $I^\gamma(L^p(\R^N))$, $\gamma >0$, $1<p<\infty$, is given in \cite[Theorem 26.8, 2)]{S-K-M}:
\begin{proposition}
Let $f$ be locally integrable and $\lim_{|x|\to\infty} f(x) = 0$. Then $f \in I^\gamma(L^p(\R^N))$, $\gamma>0$,    $1<p<\infty$, if and only if 
\begin{itemize}
\item When $1<p<\frac{N}{\gamma}$, then $f \in L^q(\R^N)$ with $q= \frac{Np}{N-\gamma p}$ and $D^\gamma f \in L^p(\R^N)$,  
\item When $p\geq \frac{N}{\gamma}$, then  
$$
(\triangle^l_h f) \in L^p(\R^N), \  \text{and there exists in} \ L^p(\R^N),  \  \lim_{\epsilon\to 0} \int_{|h|>\epsilon} |h|^{-N-\gamma} (\triangle^l_h f)(x) dh, 
$$
where $l > 2[\frac{\gamma}{2}]$. For the definition of the {\em difference operators} $(\triangle^l_h f)$ we refer to \cite[Lemma 26.5]{S-K-M}.
\end{itemize}   
Moreover, $I^\gamma(L^p(\R^N)) \cap L^r(\R^N) = L^{\gamma}_{p,r}(\R^N)$, $\gamma >0$, $1\leq p < \infty$, $1\leq r <\infty$. 
\end{proposition}

We are using the above description for $N= \gamma=2$ with $I^2(\psi) = \int_{\R^2} \ln(\frac{1}{|x-y|}) \psi(y) dy$. We also have $\mathcal{L}^p_{\gamma}(\R^2) = L^{\gamma}_{p,p}(\R^2)$.

\section{Appendix. Properties of the Riesz operators, and estimates.}\label{4.5}
We recall first the properties of the Laplace operator in $L^p(\R^N), 1<p< +\infty$ (e.g. \cite[p.37]{M-S}). The distributional Laplacian restricted to $L^p(\R^N)$; $\Delta_p$, with the domain $D(\Delta_p) = W^{2,p}(\R^N)$, has the following properties:
\begin{itemize}
\item $\Delta_p$ is a closed operator, 
\item $(-\Delta_p)$ is m-accretive, 
\item If $1<p< +\infty$, then $(\Delta_p)$ is injective and has dense domain and range. 
\end{itemize}
As a particular consequence of the last property, the set $I^2(L^p(\R^N)) = (-\Delta)^{-1}(\R^N)$ is dense in $L^p(\R^N)$ since the range $R(-\Delta) = (-\Delta)(L^p(\R^N)) = D((-\Delta)^{-1})$ is dense.

Nice description of the Riesz operators $R_j, j=1,...,N$, is given in \cite{S}. We quote here only few of its properties. The first one says that the Riesz operator
\begin{equation*}
R_j(f)(x) = \lim_{\epsilon \to 0} c_N \int_{|y| \geq \epsilon} \frac{y_j}{|y|^{N+1}} f(x-y) dy, \  j+1,2,...,N, \  c_N = \frac{\Gamma(\frac{N+1}{2})}{\pi^{(N+1)/2}},
\end{equation*}
is well defined for $f \in L^p(\R^N), 1\leq p < \infty$, and bounded as a map from $L^p(\R^N), 1\leq p < \infty$ into itself. Moreover (\cite[p. 299]{M-S}), for $f \in L^2(\R^N) \cap L^p(\R^N), 1 < p < \infty$, they have the characterization:
\begin{equation*}
R_jf = \left(-i \frac{x_j}{|x|}\hat{f}\right)^\vee, \  1\leq j\leq N.
\end{equation*}
Further, if $f \in W^{1,p}(\R^2)$, then (e.g. \cite[p.299]{M-S})
\begin{equation*}
\frac{\partial f}{\partial x_k} = -R_k (-\Delta)^{\frac{1}{2}} f,
\end{equation*}
moreover, if $f \in C^2_0(\R^N)$, then  
\begin{equation*}
\frac{\partial^2 f}{\partial x_j \partial x_k} = -R_j R_k\Delta f.
\end{equation*}

The following {\it commutator estimate} is known in the literature: \cite[p. 61]{RE}, or \cite[p. 54]{L-P} in particular case; for $\gamma >0$,
\begin{equation}\label{commutator}
\|(-\Delta)^\gamma(FG)\|_{L^p(\R^2)} \leq c\bigl[\|(-\Delta)^\gamma F\|_{L^r(\R^2)} \|G \|_{L^q(\R^2)} + \|F\|_{L^q(\R^2)} \| (-\Delta)^\gamma G\|_{L^r(\R^2)}\bigr],
\end{equation}
where $1< p < q \leq \infty$ and $\frac{1}{r} + \frac{1}{q} = \frac{1}{p}$, and the right hand side is sensible for $F$ and $G$.

Moreover, we  quote the observation (e.g. \cite[Chapter III]{S}, \cite[p. 299]{M-S}), that {\it the Riesz transforms $R_j$ are bounded operators from $L^q(\R^N)$ into $L^q(\R^N)$, $1<q<\infty$}:
\begin{equation}\label{Rieszobs}
\exists_{C>0} \forall_{\psi \in L^q(\R^N)}  \| R_j(\psi) \|_{L^q(\R^N)} \leq C\| \psi \|_{L^q(\R^N)}, \quad j=1,2,...,N.
\end{equation}
We have also the following property taken from the paper \cite[p. 12]{W1}:
\begin{equation}\label{DU}
\| |D^j u(t,\cdot)| \|_{L^q(\R^2)} \leq c \| D^j \theta(t,\cdot)\|_{L^q(\R^2)}, \  q \in (1,\infty),\, |j| \leq k.
\end{equation}

\subsection{Properties of the fractional powers operators.}
Recall first the Balakrishnan definition of fractional power of non-negative operator (e.g. \cite[p. 299]{KO}).  Let $A$ be a closed linear densely defined operator in a  Banach space $X$, such that its resolvent set contains $(-\infty,0)$ and the resolvent satisfies:
\begin{equation*}
\| \lambda(\lambda+A)^{-1}\| \leq M, \ \lambda>0.
\end{equation*}
Then, for $\eta \in (0,1)$,
\begin{equation*}
A^\eta \phi = \frac{\sin(\pi\eta)}{\pi} \int_0^\infty s^{\eta - 1} A (s + A)^{-1} \phi ds.
\end{equation*}
Note that there is another definition, through singular integrals, of the fractional powers of the $(-\Delta)^{-\alpha}$ in $L^p(\R^N)$ frequently used in the literature. See \cite[Chapter 2.2]{M-S} for the proof of equivalence of the two definitions for $1<p<\frac{N}{2\Re \alpha}$; see also \cite[section 4.3]{D-K-S1}.  

{\bf Moment inequality.} We recall here, for completeness of presentation, the {\it moment inequality} estimate valid for fractional powers of non-negative operators having zero in the resolvent. Recalling \cite[p. 98]{Ya}, we have the following version of the moment inequality; for $0 \leq \alpha < \beta < \gamma \leq 1$:
\begin{equation*}
\| A^\beta \phi\| \leq \frac{(\sin \frac{\pi(\beta-\alpha)}{\gamma -\alpha}) (\gamma -\alpha)^2}{\pi(\gamma -\beta)(\beta -\alpha)} (M +1) \| A^\gamma \phi \|^{\frac{\beta -\alpha}{\gamma -\alpha}} \| A^\alpha \phi \|^{\frac{\gamma -\beta}{\gamma -\alpha}},
\end{equation*}
where $\phi \in D(A^\gamma)$. In particular, for $\frac{1}{2} < \beta \leq 1$, this reduces to
\begin{equation*}
\| A^\beta \phi\| \leq \frac{\sin(2\pi(\beta -\frac{1}{2}))}{4 \pi(1-\beta)(\beta -\frac{1}{2})} (M +1) \| A \phi \|^{2\beta - 1} \| A^{\frac{1}{2}} \phi \|^{2 -2\beta},
\end{equation*}
for $\phi \in D(A)$. Evidently we can let $\beta \searrow \frac{1}{2}$ in the formula above. 

\section{Critical problem $\alpha=\frac{1}{2}$ in bounded domain $\Omega \subset \R^2$.}\label{addedsect} 
Solvability for small initial data. It is known from \cite{C-C-W}, that the critical equation possess smooth global solution provided the $\| \theta_0 \|_{L^\infty}$ is sufficiently small. Within our technique similar result will be now proved for the Dirichlet problem. The same will be shown also in case of the {\it space periodic boundary conditions}.

\subsection{Quasi-geostrophic equation in bounded domain $\Omega$.} Very briefly we discuss now an extension of the earlier result obtained for the Cauchy problem in $\R^2$ to the case of the Dirichlet boundary value problem for the quasi-geostrophic equation in bounded domain $\Omega \subset \R^2$: 
\begin{equation}\label{Q-gequationOmega}
\begin{split}
&\theta_t + \nabla\cdot(u\theta) + \kappa(-\Delta)^\alpha \theta = f, \quad x \in \Omega \subset \R^2,\, t>0,  \\
&\theta = 0 \  \text{on}  \  \partial \Omega,    \\ 
&\theta(0,x) = \theta_0(x),
\end{split}
\end{equation}
Note, the nonlinear term above is now written in an equivalent form. Indeed, we  have used the property, which will be checked by a direct calculation on the expression $u = (-\frac{\partial\psi}{\partial x_2}, \frac{\partial\psi}{\partial x_1})$, that for regular solutions:      
\begin{equation}
u \cdot \nabla\theta = \nabla \cdot(u\theta). 
\end{equation}  
Such equivalent form will be more suitable for getting the required estimates.    

Working with {\it sectorial positive operator} $A_\alpha := (-\Delta)^\alpha, \alpha \in [\frac{1}{2},1]$, in $\Omega$ with zero boundary condition (e.g. \cite{HE, C-D}), we rewrite equation \eqref{Q-gequationOmega} in an abstract form:
\begin{equation}\label{6.1}
\begin{split}
&\theta^\alpha_t + \nabla\cdot(u\theta^\alpha) + A_\alpha \theta^\alpha = f, \ x \in \Omega, t>0,  \\
&\theta^\alpha(0,x) = \theta_0(x).  
\end{split}
\end{equation}  

Precisely as in the case of the Cauchy problem in $\R^2$ one can get the following local existence result. 
\begin{theorem}\label{lsloc}
Let $s>1$ be fixed. Then, for $f\in H^{s}(\Omega)$ and for arbitrary $\theta_0 \in H^{2\alpha^- +s}(\Omega)\cap D((-\Delta)^{[\frac{s}{2\alpha}]+1})$, there exists in the phase space $H^{2\alpha^- +s}(\Omega)$ a unique global in time mild solution $\theta(t)$ to the subcritical problem \eqref{6.1}, $\alpha \in (\frac{1}{2},1]$. Moreover,  
\begin{equation*}
\theta \in C((0,+\infty); H^{2\alpha+s}(\Omega)) \cap C([0,+\infty); H^{2\alpha^- +s}(\Omega)), \  \theta_t \in C((0,+\infty); H^{2\gamma+s}(\Omega)),
\end{equation*}
with arbitrary $\gamma < \alpha^-$.  Furthermore, the {\em Cauchy formula} is satisfied:
\begin{equation*}
\theta(t) = e^{-A_\alpha t} \theta_0 + \int_0^t e^{-A_\alpha(t-s)} F_1(\theta(s)) ds, \  t \in [0, +\infty),
\end{equation*}
where $e^{-A_{\alpha} t}$ denotes the linear semigroup corresponding to the operator $A_\alpha := (-\Delta)^\alpha$ in $H^s(\Omega)$ (with proper number of boundary conditions when $s$ is large), and
$$
F_1(\theta) = R_2 \theta \frac{\partial \theta}{\partial x_1} - R_1 \theta \frac{\partial \theta}{\partial x_2}.  
$$
\end{theorem}

Further, a version of the Maximum Principle \eqref{maxp11} will give us global in time extendibility of that local solution, precisely as in the case of the Cauchy problem. 

Our main result, for the Dirichlet problem case, is that for small $\|\theta_0\|_{L^\infty(\Omega)}$ the solutions to subcritical problems \eqref{Q-gequationOmega} converge to the solution of the critical problem as $\alpha\to \frac{1}{2}^+$.   

In this section we denote the solution of \eqref{6.1} as $\theta^\alpha$ (the superscript is added for clarity).  Such solutions, for any $\alpha \in (\frac{1}{2},1]$, are varying continuously in $H^{1+s}(\Omega)$, $s>1$, hence also in each of the spaces $W^{1,p}(\Omega)$, $2 \leq p \leq +\infty$. In particular they fulfill the {\it uniform in $\alpha \in (\frac{1}{2}, 1]$} estimates in $L^p(\Omega)$ of Lemma \ref{MaxPrinc}, the {\it main information} allowing us to let $\alpha \to \frac{1}{2}^+$ in the equation \eqref{Q-gequationOmega}. More precisely,  we have:
\begin{equation}\label{MP}
\exists_{const >0} \forall_{\alpha \in (\frac{1}{2}, \frac{3}{4}]} \forall_{p \in [2,+\infty)}  \  \|\theta^\alpha \|_{L^\infty([0,T]; L^p(\Omega))} \leq const,
\end{equation}
with $T>0$ fixed arbitrary large. 

Note also that we will let $q \to +\infty$ in the estimate \eqref{maxp11}. Indeed, for $\theta_0, f \in L^\infty(\Omega)$, taking q-th roots in \eqref{maxp11}, letting $q \to +\infty$ and notifying that $\lim_{q\to +\infty} e^{\frac{(q-1)t}{q}} = e^t$, we obtain: 
\begin{equation}\label{5.16}
\|\theta(t,\cdot)\|_{L^\infty(\Omega)} \leq \bigl(\| \theta_0 \|_{L^\infty(\Omega)} + \|f\|_{L^\infty(\Omega)}\bigr)e^t.  
\end{equation}
Consequently $p= +\infty$ is allowed in \eqref{MP}. The last estimate is valid also for the Cauchy problem.   

To prove the convergence of solutions as $\alpha\to \frac{1}{2}^+$, consider the difference of two equations \eqref{6.1} with different $\alpha \in (\frac{1}{2}, \frac{3}{4}]$, say $\frac{1}{2} < \alpha_2 < \alpha_1\leq \frac{3}{4}$. For simplicity we will use now the notation: $\theta_i; i=1,2$, are the solutions of \eqref{6.1} (constructed in Theorem \ref{lsloc} in $H^s(\Omega), s>1$) corresponding to proper $\alpha_i$. Multiplying the above difference in $L^2(\Omega)$ by $A_{\alpha_1}^{-1}(\theta_1 -\theta_2)$ we obtain:
\begin{equation}  
\begin{split}
\int_{\Omega} (\theta_1-\theta_2)_t A_{\alpha_1}^{-1}(\theta_1 -\theta_2) dx &= -\int_{\Omega} \left(A_{\alpha_1}(\theta_1 - \theta_2) + (A_{\alpha_1} - A_{\alpha_2}) \theta_2\right) A_{\alpha_1}^{-1}(\theta_1 -\theta_2) dx   \\  
&-\int_{\Omega} \bigl(u_1 \cdot \nabla\theta_1 - u_2 \cdot \nabla\theta_2\bigr) A_{\alpha_1}^{-1}(\theta_1 -\theta_2) dx. 
\end{split}
\end{equation}  
Working in a Hilbert space $L^2(\Omega)$ is a bit simpler than in general $L^p(\Omega)$; we need not use the 'conjugate' operators like in general case, and our task is to get {\it a weak form of convergence} of the two solutions $\theta_1$ to $\theta_2$ in $H^{-\frac{1}{2}}(\Omega)$. With the use of our previous regularity knowledge of the solutions such knowledge will be extended to better norms convergence using the interpolation argument.  

We will transform the components one by one as follows.  
\begin{equation}\label{5.4a}
\int_{\Omega} (\theta_1-\theta_2)_t A_{\alpha_1}^{-1}(\theta_1 -\theta_2) dx = \frac{1}{2} \frac{d}{dt} \int_{\Omega} (A_{\alpha_1}^{-\frac{1}{2}}(\theta_1 -\theta_2))^2 dx, 
\end{equation}
\begin{equation}\label{5.5a}
-\int_{\Omega} A_{\alpha_1}(\theta_1 - \theta_2) A_{\alpha_1}^{-1}(\theta_1 -\theta_2) dx = -\|\theta_1 - \theta_2\|^2_{L^2(\Omega)},
\end{equation}
\begin{equation}\label{5.6a}
\begin{split}
-\int_{\Omega} (A_{\alpha_1} - A_{\alpha_2}) \theta_2 A_{\alpha_1}^{-1}(\theta_1 -\theta_2) dx &= -\int_{\Omega} A_{\alpha_1}^{-1}(A_{\alpha_1} - A_{\alpha_2})\theta_2 (\theta_1 -\theta_2) dx  \\
&\leq \|A_{\alpha_1}^{-1}(A_{\alpha_1} - A_{\alpha_2})\theta_2\|_{L^2(\Omega)} \|\theta_1 -\theta_2\|_{L^2(\Omega)},   
\end{split}
\end{equation} 
\begin{equation}\label{nonlinear}
\begin{split}
|\int_{\Omega} \bigl(u_1 \cdot \nabla\theta_1 - u_2 \cdot \nabla\theta_2\bigr) A_{\alpha_1}^{-1}(\theta_1 -\theta_2) dx| &= |\int_{\Omega} A_{\frac{1}{2}}^{-1}\bigl(u_1 \cdot \nabla\theta_1 - u_2 \cdot \nabla\theta_2\bigr) A_{\alpha_1-\frac{1}{2}}^{-1}(\theta_1 - \theta_2) dx|   \\
&\leq \|A_{\frac{1}{2}}^{-1}\bigl(u_1 \cdot \nabla\theta_1 - u_2 \cdot \nabla\theta_2\bigr)\|_{L^2(\Omega)} 
\|A_{\alpha_1-\frac{1}{2}}^{-1}(\theta_1 - \theta_2)\|_{L^2(\Omega)},  
\end{split} 
\end{equation} 
we leave estimating the nonlinearity for a moment.

To deal with intermediate term, we will use the following property.  
\begin{lemma}\label{new} 
Let $A$ be a {\em positive operator} in a Banach space $X$ (\cite{TR, M-S, C-D}). For arbitrary  $\phi \in X$, we have  
\begin{equation*}
\forall_{\epsilon >0} \exists_{L}  \|(I - A^{-\beta})\phi \|_{X} \leq \sin(\pi\beta) \bigl(\frac{2L(1+M)}{\pi} + L^{-1}M)\| \phi\|_X + \epsilon.    
\end{equation*}
Consequently, the left hand side tends to zero as $0 < \beta \to 0^+$.    
\end{lemma}
We will use that lemma to estimate the difference $(I - A_{\alpha_1}^{-1}A_{\alpha_2})$ when $\frac{1}{2} < \alpha_2 \leq \alpha_1 \to \frac{1}{2}^+$.         
\begin{proof}
Our task is, for fixed $\phi \in X$ and $\beta$ near $0^+$, to estimate the expression:  
\begin{equation}\label{integr}
\begin{split}
(A^{-\beta}-I)\phi &= \frac{\sin(\pi\beta)}{\pi} \int_0^\infty \lambda^{-\beta} (\lambda+A)^{-1} \phi d\lambda - \frac{\sin(\pi(1-\beta))}{\pi} \int_0^\infty \frac{\lambda^{(1-\beta)-1}}{\lambda+1} d\lambda \  \phi  \\   
&= \frac{\sin(\pi\beta)}{\pi} \int_0^\infty \lambda^{-\beta} [(\lambda+A)^{-1} \phi - \frac{1}{\lambda+1} \phi] d\lambda.
\end{split}
\end{equation}
In the estimates we are using the following properties; taken from \cite[p. 62]{M-S} equality valid for $\eta \in (0,1)$   
\begin{equation*}
\int_0^\infty \frac{\lambda^{\eta-1}}{\lambda+1} d\lambda = \frac{\pi}{\sin(\pi\eta)}, 
\end{equation*}
the simple formula:    
\begin{equation}\label{raz}
(\lambda+A)^{-1} \phi - \frac{1}{\lambda+1} \phi = \frac{1}{\lambda+1} \bigl[\lambda(\lambda+A)^{-1} \phi-\phi +(\lambda+A)^{-1} \phi\bigr],      
\end{equation}
and the two asymptotic properties of {\it non-negative operators} valid on functions $\phi \in X$ taken from \cite[Proposition 1.1.3]{M-S}:
\begin{equation}\label{dwaa}
\begin{split}
&\lim_{\lambda \to \infty} \lambda(\lambda+A)^{-1} \phi = \phi,     \\
&\lim_{\lambda\to \infty} (\lambda+A)^{-1} A\phi = 0.
\end{split}  
\end{equation}  

Returning to the proof, we split the integral in \eqref{integr} into $(0,L)$ and $(L,\infty)$ and estimate the first part,  
\begin{equation}
\begin{split}
\frac{\sin(\pi\beta)}{\pi}\| \int_0^L \lambda^{-\beta} (\frac{1}{\lambda+1} - (\lambda+A)^{-1})\phi d\lambda\|_X &\leq \frac{\sin(\pi\beta)}{\pi} \int_0^L \lambda^{-\beta}(1 + M) d\lambda \| \phi\|_X    \\
&=\frac{\sin(\pi\beta)}{\pi} \frac{L^{1-\beta}}{1-\beta} (1+M) \|\phi\|_X,
\end{split}
\end{equation}
where $L>0$ will be chosen later. Note that letting $\beta \to 0^+$ the result of the estimate above is bounded by $|\frac{\sin(\pi\beta)}{\pi}| 2L(1+M)\|\phi\|_X \to 0$, for any fixed $L>0$.

Next using \eqref{raz}, the integral over $(L, \infty)$ is, for $\phi \in X$, estimated as follows: 
\begin{equation}
\frac{\sin(\pi\beta)}{\pi} \|\int_L^\infty \lambda^{-\beta} [(\lambda+A)^{-1} \phi - \frac{1}{\lambda+1} \phi] d\lambda \|_X \leq \frac{\sin(\pi\beta)}{\pi} \int_L^\infty \frac{\lambda^{-\beta}}{\lambda+1} \| \lambda(\lambda+A)^{-1} \phi-\phi +(\lambda+A)^{-1} \phi \|_X d\lambda, 
\end{equation}
where due to \eqref{dwaa} we see that       
\begin{equation}\label{5.14}  
\|(\lambda+1)(\lambda+A)^{-1} A\phi - \phi\|_X \leq  \| \lambda(\lambda+A)^{-1} \phi-\phi\|_X + \|(\lambda+A)^{-1} \phi \|_X \leq \epsilon +\frac{M}{1+\lambda} \|\phi\|_X  \  \text{as}  \  \lambda \to \infty,
\end{equation}
$\epsilon>0$ arbitrary fixed.  Consequently we obtain
\begin{equation}
\begin{split}
\frac{\sin(\pi\beta)}{\pi}  \int_L^\infty \frac{\lambda^{-\beta}}{\lambda +1} \bigl(\epsilon +\frac{M}{1+\lambda} \|\phi\|_X) d\lambda &\leq \frac{\sin(\pi\beta)}{\pi}  \int_L^\infty \lambda^{-\beta-1} \bigl(\epsilon +\frac{M}{\lambda} \|\phi\|_X) d\lambda    \\
&\leq \frac{\sin(\pi\beta)}{\pi} \frac{L^{-\beta}}{\beta}\epsilon + \frac{\sin(\pi\beta)}{\pi} \frac{L^{-\beta-1}}{\beta +1} M \|\phi\|_X,   
\end{split}
\end{equation}    
for sufficiently large value of $L \geq 1$, as specified in \eqref{5.14}. Note that letting $\beta \to 0^+$ in the resulting estimate we have:
\begin{equation}
\frac{\sin(\pi\beta)}{\pi\beta} L^{-\beta} \epsilon + \frac{\sin(\pi\beta)}{\pi} \frac{L^{-\beta-1}}{\beta +1} M \|\phi\|_X \leq \epsilon + \sin(\pi\beta) L^{-1} M \|\phi\|_X,
\end{equation}
for chosen large value of $L$.   

For such $L$ we get a final estimate of the integral in \eqref{integr} having the form:     
\begin{equation}\label{5.17}
\| (A^{-\beta}-I)\phi\|_X \leq \frac{\sin(\pi\beta)}{\pi} \|\int_0^\infty \lambda^{-\beta} [(\lambda+A)^{-1} \phi - \frac{1}{\lambda+1} \phi] d\lambda \|_X \leq \sin(\pi\beta) \bigl(\frac{2L(1+M)}{ \pi} + L^{-1} M \bigr)\|\phi\|_X + \epsilon,    
\end{equation}
where $\epsilon >0$ was arbitrary. The right hand side of \eqref{5.17} will be made small when we let $\beta$ near $0^+$, noting $\epsilon$ was an arbitrary positive number.   
\end{proof}
 
Now we intent to estimate the nonlinearity in terms of the difference $(\alpha_1-\frac{1}{2})$. As a consequence of the property: $u\cdot \nabla\theta =\nabla \cdot(u\theta)$, in the case of the difference, we obtain
\begin{equation}\label{5.19a}
u_1 \cdot \nabla\theta_1 - u_2 \cdot\nabla\theta_2 = \nabla \cdot(u_1\theta_1 - u_2\theta_2),  
\end{equation}
so that    
\begin{equation}\label{puc}
\begin{split}
 &\|A_{\frac{1}{2}}^{-1}(u_1 \cdot \nabla\theta_1 - u_2 \cdot\nabla\theta_2)\|_{L^2(\Omega)} =  \|(-\Delta)^{-\frac{1}{2}} \nabla \cdot(u_1\theta_1 - u_2\theta_2)\|_{L^2(\Omega)}   \\
&= \|(-\Delta)^{-\frac{1}{2}} \nabla \cdot(u_1\theta_1 -u_1\theta_2 + u_1\theta_2 - u_2\theta_2)\|_{L^2(\Omega)} \leq c \bigl(\|\theta_1\|_{L^\infty(\Omega)} + \|\theta_2\|_{L^\infty(\Omega)}\bigr) \|\theta_1 - \theta_2\|_{L^2(\Omega)},   
\end{split}
\end{equation}   
where we have used \eqref{DU} with $j=0$ and boundedness of the operator $(-\Delta)^{-\frac{1}{2}} \nabla$ in $L^2(\Omega)$.

Note further that, from the formula of negative powers  
\begin{equation}
A^{-\beta} \phi = \frac{\sin(\pi \beta)}{\pi} \int_0^\infty \lambda^{-\beta} (A+\lambda I)^{-1} \phi d\lambda,
\end{equation}
by splitting the integral over $(0,L)$ and $(L,\infty)$, we get an estimate 
\[
\| A^{-\beta} \phi \|_X \leq \frac{\sin(\pi\beta)}{\pi} M \| \phi \|_X \bigl(\frac{L^{1-\beta}}{1- \beta} + \frac{L^{-\beta}}{\beta}\bigr),
\]
and then by taking $L=1$ we have, as $\beta\in (0,\frac{1}{2})$,  
\[
\| A^{-\beta} \phi \|_X \leq \frac{\sin(\pi\beta)}{\pi} M \| \phi \|_X \bigl(\frac{1}{1- \beta} + \frac{1}{\beta}\bigr)\leqslant \frac{2}{\pi}M \| \phi \|_X+ \frac{\sin(\pi\beta)}{\pi\beta} M \| \phi \|_X \leqslant (\frac{2}{\pi}+2)M \| \phi \|_X.
\]

Thus, thanks to the last estimate used for $\beta= \alpha_1-\frac{1}{2}$,     
\begin{equation}\label{5.24a}
\begin{split}
&|\int_{\Omega} \bigl(u_1 \cdot \nabla\theta_1 - u_2 \cdot \nabla\theta_2\bigr) A_{\alpha_1}^{-1}(\theta_1 -\theta_2) dx|\\
 &\leqslant c \bigl(\|\theta_1\|_{L^\infty(\Omega)} + \|\theta_2\|_{L^\infty(\Omega)}\bigr) \|\theta_1 - \theta_2\|_{L^2(\Omega)} \|(-\Delta)^{\alpha_1-\frac{1}{2}}(\theta_1 - \theta_2)\|_{L^2(\Omega)}    \\
 &\leqslant (\frac{2}{\pi}+2)M c \bigl(\|\theta_1\|_{L^\infty(\Omega)} + \|\theta_2\|_{L^\infty(\Omega)}\bigr) \|\theta_1 - \theta_2\|^2_{L^2(\Omega)}.  
\end{split}
\end{equation}

Putting estimates \eqref{5.4a}-\eqref{5.6a} and \eqref{5.24a} together, using Cauchy's and Young's inequalities, we obtain a differential inequality of the form:
\begin{equation}\label{5.15}
\begin{split}
\frac{d}{dt} &\int_{\Omega} (A_{\alpha_1}^{-\frac{1}{2}}(\theta_1 -\theta_2))^2 dx \leq  -\|\theta_1 - \theta_2\|^2_{L^2(\Omega)} + O(\alpha_1-\alpha_2)\|\theta_2\|_{L^2(\Omega)}\|\theta_1 - \theta_2\|_{L^2(\Omega)}  \\
&\quad + (\frac{2}{\pi}+2)M c \bigl(\|\theta_1\|_{L^\infty(\Omega)} + \|\theta_2\|_{L^\infty(\Omega)}\bigr) \|\theta_1 - \theta_2\|^2_{L^2(\Omega)}   \\
&\leqslant \left(-1 + (\frac{2}{\pi}+2)M c \bigl(\|\theta_1\|_{L^\infty(\Omega)} + \|\theta_2\|_{L^\infty(\Omega)}\bigr)\right) \|\theta_1 - \theta_2\|^2_{L^2(\Omega)} +c_3O(\alpha_1-\alpha_2)   \\
&\leqslant c_1 \left(-1 + (\frac{2}{\pi}+2)M c \bigl(\|\theta_1\|_{L^\infty(\Omega)} + \|\theta_2\|_{L^\infty(\Omega)}\bigr)\right) \|A_{\alpha_1}^{-\frac{1}{2}}(\theta_1 - \theta_2)\|^2_{L^2(\Omega)} + c_3O(\alpha_1-\alpha_2),
\end{split} 
\end{equation} 
where $O(\alpha_1-\alpha_2)$ comes from \eqref{5.6a} ($O(\cdot)$ means that $O(\alpha)\to 0$ as $\alpha\to 0$) and $c_1$ is an embedding constant for $L^2(\Omega) \subset H^{-\alpha_1}(\Omega)$. 

We consider the last estimate under the 'smallness hypothesis': {\it Let $\|\theta_0\|_{L^\infty(\Omega)}$ be so small, that the collected coefficient above 
\begin{equation}\label{pac}
-c_2:= c_1\left(-1 + (\frac{2}{\pi}+2)M c \bigl(\|\theta_1\|_{L^\infty(\Omega)} + \|\theta_2\|_{L^\infty(\Omega)}\bigr)\right) < 0,
\end{equation}  
is negative on $[0,T]$}.  Since $\int_{\Omega} (A_{\alpha_1}^{-\frac{1}{2}}(\theta_1(0) -\theta_2(0)))^2 dx=0$, applying the differential inequality to $y(t) := \int_{\Omega} (A_{\alpha_1}^{-\frac{1}{2}}(\theta_1(t) -\theta_2(t)))^2 dx$, we get an estimate   
\begin{equation}
y(t) \leq  \frac{c_3}{c_2} O(\alpha_1-\alpha_2) \  \text{as}  \   \alpha_1-\alpha_2 \to 0,
\end{equation}
when $\alpha_1 \searrow \frac{1}{2}$.  

Consequently, under the smallness hypotheses \eqref{pac}, we have the convergence
\begin{equation}
\|\theta_1 -\theta_2\|_{H^{-\frac{1}{2}}(\Omega)} \to 0.
\end{equation}     
Using the {\it interpolation argument} and Maximum Principle, the last convergence will be strengthened to
\begin{equation}\label{7.28}
\|\theta_1 -\theta_2\|_{H^{-\epsilon}(\Omega)} \leq c \|\theta_1 -\theta_2\|^{2\epsilon}_{H^{-\frac{1}{2}}(\Omega)} \|\theta_1 -\theta_2\|^{1- 2\epsilon}_{L^2(\Omega)} 
\leq c \|\theta_1 -\theta_2\|^{2\epsilon}_{H^{-\frac{1}{2}}(\Omega)} \bigl(\|\theta_1\|^{1- 2\epsilon}_{L^2(\Omega)} + \| \theta_2\|^{1- 2\epsilon}_{L^2(\Omega)}\bigr)  \to 0,  
\end{equation}
where $\epsilon>0$ is an arbitrary small number. Therefore the convergence holds in each {\it negative order} space, but not in $H^0(\Omega)=L^2(\Omega)$.   
 
We note also, that since $H^{\frac{1}{2}}(\Omega) \subset L^4(\Omega), N=2$, then by duality $L^{\frac{4}{3}^+}(\Omega) \subset L^{\frac{4}{3}}(\Omega) \subset H^{-\frac{1}{2}}(\Omega)$, and the following convergence as $\alpha_1 \searrow \frac{1}{2}$ holds:  
\begin{equation}\label{5.29}    
\|\theta_1 -\theta_2\|_{L^{\frac{4}{3}^+}(\Omega)} \leq c \|\theta_1 -\theta_2\|_{H^{-\frac{1}{2}}(\Omega)}^\mu \|\theta_1 -\theta_2\|_{L^2(\Omega)}^{1 -\mu} \to 0, 
\end{equation}
where $\mu \in (0,1)$.   

\begin{remark}  
Consider now the case when the smallness hypotheses \eqref{pac} is violated.  Since we have Maximum Principle valid also in $L^4(\Omega)$, we will divide the first right hand side term in \eqref{5.15} into two parts and estimate the result using Nirenberg-Gagliardo inequality, as follows:
\begin{equation}
\begin{split}
\frac{d}{dt} &\int_{\Omega} (A_{\alpha_1}^{-\frac{1}{2}}(\theta_1 -\theta_2))^2 dx \leq  -\|\theta_1 - \theta_2\|^2_{L^2(\Omega)} + O(\alpha_1-\alpha_2)\|\theta_2\|_{L^2(\Omega)}\|\theta_1 - \theta_2\|_{L^2(\Omega)}  \\
\quad &+(\frac{2}{\pi}+2)M c \bigl(\|\theta_1\|_{L^\infty(\Omega)} + \|\theta_2\|_{L^\infty(\Omega)}\bigr) \|\theta_1 - \theta_2\|^2_{L^2(\Omega)}   \\
\leqslant &-\frac{1}{2}\|\theta_1 - \theta_2\|^2_{L^2(\Omega)} + (\frac{2}{\pi}+2)M c \bigl(\|\theta_1\|_{L^\infty(\Omega)} + \|\theta_2\|_{L^\infty(\Omega)}\bigr) \|\theta_1 - \theta_2\|^2_{L^2(\Omega)} + c_3O(\alpha_1-\alpha_2)   \\
\leqslant &-\frac{1}{2}\|\theta_1 - \theta_2\|^2_{L^2(\Omega)} +  (\frac{2}{\pi}+2)M c \bigl(\|\theta_1\|_{L^\infty(\Omega)} + \|\theta_2\|_{L^\infty(\Omega)}\bigr) \|\theta_1 - \theta_2\|_{L^4(\Omega)} \|A_{\alpha_1}^{-\frac{1}{2}}(\theta_1 -\theta_2)\|_{L^2(\Omega)}   \\
&+ c_3O(\alpha_1-\alpha_2). 
\end{split} 
\end{equation} 
We obtain, for  $y(t) := \int_{\Omega} (A_{\alpha_1}^{-\frac{1}{2}}(\theta_1(t) -\theta_2(t)))^2 dx$, a differential inequality of the form:  
\begin{equation}
\begin{split}  
&y'(t) \leq -\frac{1}{2} y(t) + C_2 \sqrt{y(t)} + c_3O(\alpha_1-\alpha_2),  \\
&y(0) = 0,
\end{split}
\end{equation}
with positive constants $C_2, c_3$. Unfortunately, its solution need not tend to zero when $\alpha_1 \searrow \frac{1}{2}$. It is only bounded by the argument $y_{max}$, the  positive zero of the right hand side function $f(y):= -\frac{1}{2}y +C_2 \sqrt{y} +c_3O(\alpha_1-\alpha_2)$, $y \geq 0$. This gives the bound $y(t) \leq y_{max}$, valid for all positive times.  
\end{remark}

\subsection{Passing to the limit in  equation \eqref{6.1}.} Having proved the convergence $\theta^\alpha\to \theta$ in $H^{-\frac{1}{2}}(\Omega)$ together  with its consequences \eqref{7.28}, \eqref{5.29}, we consider equation \eqref{6.1}:
\begin{equation*}  
\begin{split}
&\theta^\alpha_t + \nabla \cdot(u^\alpha\theta^\alpha) + A_\alpha \theta^\alpha = f, \ x \in \Omega, t>0,  \\
&\theta^\alpha(0,x) = \theta_0(x),
\end{split}
\end{equation*}
We apply to \eqref{6.1} the operator $A_\alpha^{-1}= (-\Delta)^{-\alpha}$  and will {\it reduce the two operators} as in the proof of convergence above. Multiplying the result by a 'test function' $\eta \in H^1(\Omega)$ (note, $H^1(\Omega) \subset L^4(\Omega), N=2$), we obtain  
\begin{equation}
\int_{\Omega} A_\alpha^{-1}\theta^\alpha_t \eta dx + \int_\Omega A_\alpha^{-1} \nabla \cdot(u^\alpha\theta^\alpha) \eta dx + \int_\Omega \theta^\alpha \eta dx = \int_\Omega A_\alpha^{-1}f \eta dx,
\end{equation}
or equivalently,   
\begin{equation}
\begin{split}
\int_{\Omega} (-\Delta)^{-\frac{\alpha}{2}}\theta^\alpha_t (-\Delta)^{-\frac{\alpha}{2}}\eta dx + \int_{\Omega} (-\Delta)^{-(\alpha-\frac{1}{2})} (-\Delta)^{-\frac{1}{2}} \nabla \cdot (u^\alpha\theta^\alpha) \eta dx &+ \int_\Omega \theta^\alpha \eta dx   \\
&= \int_\Omega (-\Delta)^{-\frac{\alpha}{2}}f (-\Delta)^{-\frac{\alpha}{2}}\eta dx.  
\end{split}  
\end{equation}
Now, for the nonlinear term, 
\begin{equation}  
\int_{\Omega} (-\Delta)^{-(\alpha -\frac{1}{2})} (-\Delta)^{-\frac{1}{2}} \nabla \cdot (u^\alpha\theta^\alpha) \eta dx = \int_{\Omega} (-\Delta)^{-\frac{1}{2}} \nabla \cdot (u^\alpha\theta^\alpha) (-\Delta)^{-(\alpha -\frac{1}{2})}\eta dx. 
\end{equation} 
Noting boundedness of the operator $(-\Delta)^{-\frac{1}{2}}\nabla$ in $L^{\frac{4}{3}^+}(\Omega)$ and \eqref{5.29}, we justify that 
\begin{equation}
\begin{split}
\|u^\alpha \theta^\alpha - u\theta\|_{L^{\frac{4}{3}^+}(\Omega)} &\leq \|u^\alpha (\theta^\alpha - \theta) \|_{L^{\frac{4}{3}^+}(\Omega)} + \|(u^\alpha - u) \theta\|_{L^{\frac{4}{3}^+}(\Omega)}   \\
&\leq \|u^\alpha \|_{L^\infty(\Omega)} \|\theta^\alpha - \theta\|_{L^{\frac{4}{3}^+}(\Omega)} + \|u^\alpha - u\|_{L^{\frac{4}{3}^+}(\Omega)} \|\theta\|_{L^\infty(\Omega)} \to 0,
\end{split}
\end{equation}
and, thanks to Lemma \ref{new}, the convergence   
\begin{equation}
\begin{split}
\int_{\Omega} (-\Delta)^{-(\alpha -\frac{1}{2})} &(-\Delta)^{-\frac{1}{2}} \nabla \cdot (u^\alpha\theta^\alpha) \eta dx  = \int_{\Omega} (-\Delta)^{-\frac{1}{2}} \nabla \cdot (u^\alpha\theta^\alpha) (-\Delta)^{-(\alpha -\frac{1}{2})} \eta dx    \\
&\to \int_{\Omega} (-\Delta)^{-\frac{1}{2}} \nabla \cdot (u\theta) \eta dx = <(-\Delta)^{-\frac{1}{2}} \nabla \cdot(u\theta), \eta>_{L^{\frac{4}{3}^+}(\Omega),L^{4^-}(\Omega)}, 
\end{split} 
\end{equation}
as $\alpha \searrow \frac{1}{2}$. We will also write:
\begin{equation}
\int_\Omega \theta^\alpha \eta dx = \int_\Omega (-\Delta)^{-\frac{1}{2}}\theta^\alpha (-\Delta)^{\frac{1}{2}}\eta dx \to \int_\Omega (-\Delta)^{-\frac{1}{2}}\theta (-\Delta)^{\frac{1}{2}}\eta dx,
\end{equation}
to allow letting $\alpha \searrow \frac{1}{2}$ in that term. In the term containing time derivative $\theta^\alpha_t$  we can pass to the limit in the sense of 'scalar distributions' (see \cite{Li}). Indeed, since the approximating solutions $\theta^\alpha$ satisfy, in particular, 
\begin{equation}
\theta^\alpha \in L^\infty(0,T;L^2(\Omega)), \theta_t^\alpha \in L^2(0,T;L^2(\Omega)), \    \alpha \in (\frac{1}{2},\frac{3}{4}], 
\end{equation}
then by \cite[Lemma 1.1, Chapt.III]{TE1} 
\begin{equation}
\forall_{\eta \in L^2(\Omega)} <\theta_t^\alpha, \eta> = \frac{d}{dt}<\theta^\alpha,\eta> \to \frac{d}{dt}<\theta,\eta>,
\end{equation}
the derivative $\frac{d}{dt}$ and the convergence are in $\mathcal{D}'(0,T)$ (space of the 'scalar distributions'). Consequently, 
\begin{equation}
\int_{\Omega} A_\alpha^{-1}\theta^\alpha_t \eta dx = \int_{\Omega} \theta^\alpha_t A_\alpha^{-1}\eta dx = \frac{d}{dt}<\theta^\alpha,A_\alpha^{-1}\eta> \to \frac{d}{dt}<\theta,A_{\frac{1}{2}}^{-1}\eta>
\end{equation}
in $\mathcal{D}'(0,T)$ as $\alpha \searrow \frac{1}{2}$. 

Collecting all the limits together, we find the form of limit {\it critical equation}:    
\begin{equation}
\frac{d}{dt}<\theta,(-\Delta)^{-\frac{1}{2}}\eta> + <(-\Delta)^{-\frac{1}{2}} \nabla \cdot(u\theta), \eta>_{L^{\frac{4}{3}^+}(\Omega),L^{4^-}(\Omega)} + \int_\Omega (-\Delta)^{-\frac{1}{2}}\theta (-\Delta)^{\frac{1}{2}}\eta dx = \int_\Omega f (-\Delta)^{-\frac{1}{2}}\eta dx.
\end{equation}

\bigskip 
{\bf Acknowledgement.} T.D. and M.B.K. are supported by NCN grant DEC-2012/05/B/ST1/00546 (Poland). C.S. was supported by the NSFC Grants 11031003, 11171028,
lzujbky-2012-10 and NCET-11-0214.

\end{document}